\documentclass[10pt, twoside,leqno,twocolumn]{article}
\usepackage[x11names]{xcolor}
\usepackage[letterpaper]{geometry}
\usepackage{ltexpprt}
\usepackage{savesym}
\savesymbol{property} 
\usepackage{breakurl}

\usepackage{cite}
\usepackage[subtle]{savetrees}

\newif\ifconference
\conferencefalse

\def\BibTeX{{\rm B\kern-.05em{\sc i\kern-.025em b}\kern-.08emT\kern-.1667em\lower.7ex\hbox{E}\kern-.125emX}}

\usepackage[nosigalt]{supertech-acm}
\usepackage{setup}
\usepackage{fancyvrb}

\notesfalse

\begin{document}

\title{Cilkmem: Algorithms for Analyzing the Memory \\High-Water Mark of
Fork-Join Parallel Programs\thanks{MIT Computer Science and Artificial Intelligence Laboratory. Supported by NSF Grants CCF 1314547 and CCF 1533644. William Kuszmaul is supported by a  Fannie \& John Hertz Foundation Fellowship; and by a NSF GRFP Fellowship.}}
\author{Tim Kaler \and William Kuszmaul \and Tao B. Schardl \and Daniele Vettorel}
\date{}


\renewcommand{\subheading}[1]{\vspace{.2 cm} \noindent\textbf{#1}. }

\maketitle

\begin{abstract}\small\baselineskip=9pt

Software engineers designing recursive fork-join programs destined to run on
massively parallel computing systems must be cognizant of how their program's
memory requirements scale in a many-processor execution. Although tools exist
for measuring memory usage during one particular execution of a parallel
program, such tools cannot bound the worst-case memory usage over all possible
parallel executions. 

This paper introduces Cilkmem, a tool that analyzes the execution of a
deterministic Cilk program to determine its $p$-processor memory
high-water mark (MHWM), which is the worst-case memory usage of the
program over \emph{all possible} $p$-processor executions. Cilkmem
employs two new algorithms for computing the $p$-processor MHWM\@. The
first algorithm calculates the exact $p$-processor MHWM in
$O(T_1 \cdot p)$ time, where $T_1$ is the total work of the program.
The second algorithm solves, in $O(T_1)$ time, the approximate
threshold problem, which asks, for a given memory threshold $M$,
whether the $p$-processor MHWM exceeds $M/2$ or whether it is
guaranteed to be less than~$M$. Both algorithms are memory efficient,
requiring $O(p \cdot D)$ and $O(D)$ space, respectively, where $D$ is
the maximum call-stack depth of the program's execution on a single
thread.

Our empirical studies show that Cilkmem generally exhibits low overheads.
Across ten application benchmarks from the Cilkbench suite, the exact algorithm
incurs a geometric-mean multiplicative overhead of $1.54$ for $p=128$, whereas
the approximation-threshold algorithm incurs an overhead of~$1.36$ independent
of~$p$. In addition, we use Cilkmem to reveal and diagnose a previously unknown
issue in a large image-alignment program contributing to unexpectedly high
memory usage under parallel executions.

\end{abstract}


\secput{intro}{Introduction}

To design a recursive fork-join parallel program\footnote{When we talk
  about fork-join parallelism throughout this paper, we mean recursive
  fork-join parallelism.}, such as a Cilk program, to run on massively
parallel computing systems, software engineers must assess how their
program's memory requirements scale in a many-processor execution.
Many tools have been developed to observe a program execution and
report its maximum memory
consumption 
(e.g.,~\cite{VTune19, SchifferKe14, QuinlanKe17, Oracle10,
  NethercoteSe07}). But these tools can only ascertain the memory
requirements of the one particular execution of the program that they
observe.  For parallel programs, whose memory requirements can depend
on scheduling decisions that vary from run to run, existing tools are
unable to provide bounds on the maximum amount of memory that might be
used during future program executions\footnote{In this paper, when we
  consider executions of a program, we shall assume a fixed input to
  the program, including fixed seeds to any pseudorandom number
  generators the program might use.}. This paper studies the problem
of computing the $p$-processor \defn{memory high-water mark} (MHWM) of
a parallel program, which measures the worst-case memory consumption
of \emph{any} $p$-processor execution. We introduce Cilkmem, an
efficient dynamic-analysis tool that measures the MHWM of a Cilk
program for an arbitrary number of processors~$p$.

Computing the MHWM of an arbitrary parallel program is a theoretically
difficult problem. In the special case where a program's allocated memory is
freed immediately, without any intervening parallel control structure,
computing the MHWM corresponds to finding a solution to the \defn{poset chain
optimization problem} \cite{ShumTr96, CameronEd79, Shum90}. The poset chain
optimization problem is well understood theoretically, and the fastest known
algorithms run in (substantial) polynomial time using techniques from linear
programming \cite{ShumTr96}.  A direct application of these algorithms to
compute the MHWM of a parallel program would require computation that is
polynomially large in the execution time of the original program.

Many dynamic-analysis tools (e.g., \cite{FengLe99, HeLeLe10,
  SchardlKuLe15, UtterbackAgFi16, YogaNa17}) have been developed that
exploit structural properties of fork-join programs to analyze a
program efficiently.  Specifically, these tools often leverage the
fact that the execution of a fork-join program can be modeled as a
series-parallel \defn{computation DAG} (directed acyclic graph)
\cite{BlumofeLe99, FengLe99}, where the edges model executed
instructions, and the vertices model parallel-control dependencies.

But even when restricted to series-parallel DAGs,
computing the $p$-processor MHWM efficiently is far from
trivial. Identifying the worst-case memory requirement of a
$p$-processor execution involves
solving an optimization problem that sparsely assigns a finite number
of processors to edges in the program's computation DAG\@.  Such a
computation DAG can be quite large, because of the liberal nature in
which fork-join programs expose logically parallel operations.
Moreover, whereas the poset chain optimization problem assumes that
memory is freed immediately after being allocated, fork-join programs
can free memory at any point that serially follows the
allocation. Efficient solutions for this optimization problem are not
obvious, and seemingly require a global view of the program's entire
computation DAG\@.  To obtain such a view, a tool would need to store
a complete trace of the computation for offline processing and incur
the consequent time and space overheads.

This work shows, however, that it is possible not only to compute the
$p$-processor MHWM \emph{efficiently} for a fork-join program, but
also to do so in an \emph{online} fashion, without needing to store
the entire computation DAG\@.  Specifically, we provide an online
algorithm for computing the \emph{exact} $p$-processor MHWM in
$O(T_1 \cdot p)$ time, where $T_1$ is the total work of the program.
We also examine the \defn{approximate threshold problem}, which asks,
for a given memory threshold $M$, whether the $p$-processor MHWM
exceeds $M/2$ or whether it is guaranteed to be less than~$M$.  We
show how to solve the approximate threshold problem in $O(T_1)$ time
using an online algorithm.
Both of these algorithms are space efficient, requiring $O(p \cdot D)$
and $O(D)$ space, respectively, where $D$ is the maximum call-stack
depth of the program's execution on a single thread.

\subsection{Memory Consumption of Fork-Join Programs}

Let us review the fork-join parallel programming model and see how
scheduling can cause a fork-join program's memory consumption to vary
dramatically.

Recursive fork-join parallelism, as supported by parallel programming
languages including dialects of Cilk \cite{FrigoLeRa98, Leiserson10,
  IntelCilkPlus10}, Fortress \cite{AllenChHa+08}, Kokkos
\cite{EdwardsTrSu14}, Habanero \cite{BarikBuCa09}, Habanero-Java
\cite{CaveZhSh11}, Hood \cite{BlumofePa98}, HotSLAW \cite{MinIaYe11},
Java Fork/Join Framework \cite{Lea00}, OpenMP \cite{OpenMP08,
  AyguadeCoDu09}, Task Parallel Library \cite{LeijenHa07}, Threading
Building Blocks (TBB) \cite{Reinders07}, and
X10~\cite{CharlesGrSa+05}, has emerged as a popular
parallel-programming model.  In this model, subroutines can be spawned
in parallel, generating a series-parallel computation DAG of
fine-grained tasks.  The synchronization of tasks is managed ``under
the covers'' by the runtime system, which typically implements a
randomized work-stealing scheduler~\cite{BlumofeLe99, FrigoLeRa98,
  AroraBlPl98, BlumofeJoKu96}.  Constructs such as \code{parallel_for}
can be implemented as syntactic sugar on top of the fork-join model.
As long as the parallel program contains no determinacy races
\cite{FengLe99} (also called general races~\cite{NetzerMi92}), the
program is \defn{deterministic}, meaning that every program execution
on a given input performs the same set of operations, regardless of
scheduling.




\begin{figure}[t]
  \begin{minipage}[t]{1.0\linewidth}
    \begin{codebox}
      \Procname{$\proc{MemoryExplosion}(n)$}
      \li \CLRSIf $n > 1$
          \Then
      \li   \CilkSpawn $\proc{MemoryExplosion}(n-1)$ \lilabel{memoryexplosion:spawn}
          \End
      \li $b \gets \proc{malloc}(1)$ \lilabel{memoryexplosion:malloc}
      \li \CilkSync \lilabel{memoryexplosion:sync}
      \li $\proc{free}(b)$
      \li \Return
    \end{codebox}
  \end{minipage}
  \caption{Example Cilk program whose heap-memory usage can increase
    dramatically depending on how the program is scheduled.}
  \label{fig:memoryexplosion}
\end{figure}

Even a simple fork-join program can exhibit dramatic and unintuitive
changes in memory consumption,\footnote{This work focuses on
  heap-memory consumption.  In contrast, the Cilk runtime system is
  guaranteed to use stack space efficiently~\cite{BlumofeLe99}.}
based on how the program is scheduled on $p$ processors.  Consider,
for example, the Cilk subroutine \proc{MemoryExplosion} in
\figref{memoryexplosion},\footnote{Similar examples can be devised for
  other task-parallel programming frameworks.}\ which supports
parallel execution using the keywords \CilkSpawn and \CilkSync.  The
\CilkSpawn keyword on \liref{memoryexplosion:spawn} allows the
recursive call to $\proc{MemoryExplosion}(n-1)$ to execute in parallel
with the call to $\proc{malloc}(1)$ on \liref{memoryexplosion:malloc},
which allocates $1$ byte of heap memory.  The \CilkSync on
\liref{memoryexplosion:sync} waits on the spawned recursive call to
\proc{MemoryExplosion} to return before proceeding; if a thread
reaches the \CilkSync, and the recursive call to
\proc{MemoryExplosion} has not yet completed, then the thread can be
rescheduled to make progress elsewhere in the program.

Cilk's randomized work-stealing scheduler \cite{BlumofeLe99} schedules
the parallel execution of \proc{MemoryExplosion} as follows.  When a
Cilk worker thread encounters the \CilkSpawn statement on
\liref{memoryexplosion:spawn}, it immediately executes the recursive
call to $\proc{MemoryExplosion}(n-1)$.  If another worker thread in
the system has no work to do, it becomes a \defn{thief} and can
\defn{steal} the continuation of this parallel recursive call, on
\liref{memoryexplosion:malloc}.


Because of Cilk's scheduler, the memory consumption of
\proc{MemoryExplosion} can vary dramatically and nondeterministically
from run to run, even though \proc{MemoryExplosion} is deterministic.
When run on a single processor, the \CilkSpawn and \CilkSync
statements effectively act as no-ops.  Therefore,
\proc{MemoryExplosion} uses at most $1$ byte of heap memory at any
point in time, because each call to \proc{malloc} is followed by a
call to \proc{free} almost immediately thereafter.  When run on $2$
processors, however, the memory consumption of \proc{MemoryExplosion}
can increase dramatically, depending on scheduling.  While one worker
is executing \liref{memoryexplosion:spawn}, a thief can steal the
execution of \liref{memoryexplosion:malloc} and allocate $1$ byte of
memory before encountering the \CilkSync on
\liref{memoryexplosion:sync}.  The thief might then return to work
stealing, only to find another execution of
\liref{memoryexplosion:malloc} to steal, repeating the process.  As a
result, the heap-memory consumption of $\proc{MemoryExplosion}(n)$ on
two or more Cilk workers can vary from run to run between $1$ byte and
$n$ bytes, depending on scheduling happenstance.

We remark that the sequence of scheduling events that result in
$\proc{MemoryExplosion}(n)$ using a large amount of heap-memory is not
  pathological. In particular, if one models each of the two workers
  as being able to perform one operation every $O(1)$ cycles, then
  $\proc{MemoryExplosion}(n)$ is guaranteed to use heap-memory
    $\Theta(n)$ on two processors.




\subsection{Algorithms for Memory High-Water Mark} 

This paper presents algorithms for
computing the $p$-processor MHWM of a program with a series-parallel
computation DAG, and in particular, of a deterministic parallel Cilk
program $\mathcal{P}$. Let $G = (V, E)$ be the computation DAG for
$\mathcal{P}$, and suppose each edge of $G$ is annotated with the
allocations and frees within that edge.


\secref{exact} presents a simple offline algorithm for computing the
exact $p$-processor MHWM of the parallel program $\mathcal{P}$, given
the DAG~$G$.
A straightforward analysis of the exact algorithm would suggest that
it runs in time $O(T_1 \cdot p^2)$, where $T_1$ is the $1$-processor
running time of the program~$\mathcal{P}$. By performing an amortized
analysis over the parallel strands of the program, we show that a
slightly modified version of the algorithm actually achieves a running
time of $O(T_1 \cdot p)$.

Explicitly storing the DAG $G$ can be impractical for large
programs~$\mathcal{P}$. \secref{mem} presents a combinatorial
restructuring of the exact algorithm that computes the MHWM in an
\emph{online} fashion, meaning that the algorithm runs as
instrumentation on (a single-threaded execution of) the program
$\mathcal{P}$.  The online exact algorithm introduces at most $O(p)$
time and memory overheads when compared to a standard single-threaded
execution of~$\mathcal{P}$.  In particular, the algorithm runs in time
$O(T_1 \cdot p)$ and uses at most $O(p \cdot D)$ memory, where $T_1$
is the $1$-processor running time of the program, and $D$ is the
maximum call-stack depth of the program's execution on a single
thread. The simple amortization argument used for the offline
algorithm does not apply to the more subtle structure of the online
algorithm.  Instead, we employ a more sophisticated amortized
analysis, in which subportions of the graph are assigned sets of
leader vertices, and the algorithm's work is charged to the leader
vertices in such a way that no vertex receives more than $O(p)$
charge.

The two exact algorithms for computing the $p$-processor MHWM have the
additional advantage that they actually compute each of the
$i$-processor MHWM's for $i = 1, \ldots, p$.  Thus a user can determine
the largest $i \le p$ for which the $i$-processor MHWM is below some
threshold~$M$.

We also consider the approximate-threshold version of the
$p$-processor MHWM problem.  Here, one is given a number of processors
$p$ and a memory threshold $M$, and wishes to determine whether $p$
processors are at risk of coming close to running out of memory while
executing on a system with memory~$M$.  Formally, an
approximate-threshold algorithm returns a value of $1$ or $0$, where
$1$ indicates that the $p$-processor MHWM is at least $M / 2$, and
$0$ indicates that the $p$-processor MHWM is bounded above by~$M$.

\secref{approx} presents a strictly-linear time online algorithm for
the approximate-threshold problem, running in time $O(T_1)$. The
independence of the running time from $p$ means that the algorithm can
be used for an arbitrarily large number of processors $p$ while still
having a linear running time.  This property can be useful for either
understanding the limit properties of a program (i.e., behavior for
very large $p$), or the behavior that a program will exhibit on a very
large machine.  The algorithm is also memory efficient.  In
particular, the memory usage of the algorithm never exceeds $O(D)$,
where $D$ is the maximum call-stack depth of the program's serial
execution.


A key technical idea in the approximate-threshold algorithm is a lemma
that relates the $p$-processor high-water mark to the
\emph{infinite-processor} MHWM taken over a restricted set of parallel
execution states known as ``robust antichains''. The
infinite-processor MHWM over robust antichains can then be computed in
strictly linear time via a natural recursion.  To obtain an online
algorithm, we introduce the notion of a ``stripped robust antichain''
whose combinatorial properties can be exploited to remove dependencies
between non-adjacent subproblems in the recursive algorithm.

\subsection{The Cilkmem Tool}

We introduce the Cilkmem dynamic-analysis tool, which implements the
online algorithms to measure the $p$-processor memory high-water mark
of a deterministic parallel Cilk program.

Both of Cilkmem's algorithms run efficiently in practice.  We
implemented Cilkmem using the CSI framework for compiler
instrumentation \cite{SchardlDeDo17} embedded in the Tapir/LLVM
compiler~\cite{SchardlMoLe17}.  In \secref{eval}, we measure the
efficiency of Cilkmem on a suite of \fitb{ten} Cilk application
benchmarks. Cilkmem introduces only a small overhead for most of the
benchmarks. For example, the geometric-mean multiplicative overhead
across the ten benchmarks is~$1.54$, to compute the MHWM exactly for
$p = 128$, and~$1.36$, to run the approximate-threshold algorithm. For
certain benchmarks with very fine-grained parallelism, however, the
overhead can be substantially larger (although still bounded by the
theoretical guarantees of the algorithms). We find that for these
applications, the strictly-linear running time of the
approximate-threshold algorithm provides substantial performance
benefits, allowing computations to use arbitrarily large values of
$p$ with only small constant-factor overhead.

In addition to measuring Cilkmem's performance overhead, we use
Cilkmem to analyze a big-data application, specifically, an
image-alignment program \cite{KalerWhWo19} used for brain
connectomics~\cite{MatveevMeSa17}.  \secref{eval} describes how, for
this application, Cilkmem reveals a previously unknown issue
contributing to unexpectedly high memory usage under parallel
executions.

\subsection{Outline}
The remainder of the paper is structured as follows.  \secref{problemstatement}
formalizes the problem of computing the $p$-processor MHWM in terms of
antichains in series-parallel DAGs.  \secref{exact} presents the $O(T_1 \cdot
p)$-time exact algorithm, and \secref{mem} extends this to an online algorithm.
\secref{approx} presents an online linear-time algorithm for the
approximate-threshold problem. The design and analysis of the online
approximate-threshold algorithm is the most technically sophisticated part of
the paper.  \secref{eval} discusses the implementation Cilkmem, and evaluates
its performance.  \secref{relatedwork} discusses related work, and
\secref{concl} concludes with directions for future work.


\secput{problemstatement}{Problem Formalization}

This section formalizes the problem of computing the $p$-processor
memory high-water mark of a parallel program.

\subheading{The DAG model of multithreading}
Cilk programs express logical recursive fork-join parallelism through
spawns and syncs. A \defn{spawn} breaks a single thread into two
threads of execution, one of which is logically a new child thread,
while the other is logically the continuation of the original
thread. A \defn{sync} by a thread $t$, meanwhile, joins thread $t$
with the completion of all threads spawned by $t$, meaning the next
continuation of $t$ occurs only after all of its current child threads
have completed.

An execution of a Cilk program can be modeled as a computation DAG
$G = (V,E)$.  Each directed edge represents a \defn{strand}, that is,
a sequence of executed instructions with no spawns or syncs.  Each
vertex represents a spawn or a sync. 

The DAG $G$ is a \defn{series-parallel DAG} \cite{FengLe99}, which
means that $G$ has two distinguished vertices --- a \defn{source}
vertex, from which one can reach every other vertex in $G$, and
a \defn{sink} vertex, which is reachable from every other vertex in
$G$ --- and can be constructed by recursively combining pairs of
series-parallel DAGs using series and parallel combinations.
A \defn{series combination} combines two DAGs $G_1$ and $G_2$ by
identifying the sink vertex of $G_1$ with the source vertex of~$G_2$.
A \defn{parallel combination} combines two DAGs $G_1$ and $G_2$ by
identifying their source vertices with each other and their sink
vertices with each other.  We shall refer to any DAG used in a series
or parallel combination during the recursive construction of $G$ as
a \defn{component} of~$G$.  Although the recursive structure of
series-parallel DAGs suggests a natural recursive framework for
algorithms analyzing the DAG, \secref{mem} describes how a more
complicated framework is needed to analyze series-parallel DAGs in an
online fashion.

The structure of $G = (V, E)$ induces a poset on the edges $E$, in
which $e_1 < e_2$ if there is a directed path from $e_1$ to $e_2$. A
collection of distinct edges $(e_1, e_2, \ldots, e_q)$ form an \defn{antichain}
if there is no pair $e_i, e_j$ such that $e_i < e_j$. Note that edges
form an antichain if and only if there is an execution of the
corresponding parallel program in which those edges at some point run
in parallel.

\subheading{The $p$-processor memory high-water mark}
To analyze potential memory usage, we model the computation's memory
allocations and frees (deallocations) in the DAG $G$ using two
weights, $m(e)$ and $t(e)$, on each edge~$e$.  The weight $m(e)$,
called the \defn{edge maximum}, denotes the high-water mark of memory
usage at any point during the execution of $e$ when only the
allocations and frees within $e$ are considered.  The edge maximum
$m(e)$ is always non-negative since, at the start of the execution
of an edge $e$, no allocations or frees have been performed, and thus
the (local) memory usage is zero.  The weight $t(e)$, called the
\defn{edge total}, denotes the sum of allocations minus frees over the
entire execution of the edge.  In contrast to $m(e)$, an edge total
$t(e)$ can be negative when memory allocated previously in the program
is freed within~$e$.

The $p$-processor memory high-water mark is determined by the memory
requirements of all antichains of length $p$ or less in the
computation DAG~$G$. We define the \defn{water mark} $W(A)$ of an
antichain $A = (e_1, \ldots, e_q)$ to be the maximum amount of memory that
could be in use on a $q$-processor system that is executing the edges
$e_1, \ldots, e_q$ concurrently.  The \defn{$p$-processor high-water
  mark} $\HH_{p}(G)$ is the maximum water mark over antichains of
length $p$ or smaller\footnote{This definition of water mark makes no
  assumption about the underlying scheduling algorithm.  In
  particular, when a thread spawns, we make no assumptions as
  to which subsequent strand the scheduler will execute first.}:
\begin{equation}
\HH_p(G) =
\max_{\substack{(e_1, \ldots, e_q) \in \mathcal{A}, \\ \ q \le p}} W(e_1, \ldots, e_q),
\label{eq:hwmark}
\end{equation}
where $\mathcal{A}$ is the set of antichains in $G$.



\subheading{Memory water mark of an antichain} The water mark $W(A)$
of an antichain $A = (e_1,\ldots, e_q)$ is the sum of two quantities
$W(A) = W_1(A) + W_2(A)$.

The quantity $W_1(A)$ consists of the contribution to the
water mark from the edges $e_1, \ldots, e_q$ and from all the edges $e
\in G$ satisfying $e < e_i$ for some $i$:

\begin{equation}
W_1(A) = \sum_{e_i \in A} m(e_i) + \sum_{e \in E, e < e_i \text{ for some } e_i \in A} t(e).
\end{equation} 

The quantity $W_2(A)$ counts the contribution to the water mark of
what we call \defn{suspended parallel components}. If the
series-parallel construction of $G$ combines two subgraphs $G_1$ and
$G_2$ in parallel, we call them \defn{partnering parallel
  components} of $G$. Consider two partnering parallel
components $G_1$ and $G_2$, and suppose that $G_2$ contains at
least one edge from the antichain $A$, while $G_1$ does not. Then
there are two options for a parallel execution in which processors are
active in the edges of $A$: either (1) the parallel component $G_1$
has not been executed at all, or (2) the parallel component $G_1$
has been executed to completion and is \defn{suspended} until its
partner parallel component completes. In the latter case, $G_1$
will contribute $\sum_{e \in G_1} t(e)$ to the water mark of~$A$. If this
sum, which is known as $G_1$'s \defn{edge sum}, is positive, then we
call $G_1$ a \defn{companion component} to the antichain~$A$. The
quantity $W_2(A)$ counts the contribution to the water mark of edges in
companion components. That is, if $\mathcal{G}$ is the set of
companion components to $A$, then

\begin{equation}
W_2(A) = \sum_{H \in \mathcal{G}} \sum_{e \in H} t(e).
\label{eqw2}
\end{equation}

Note that the companion components of $A$ are guaranteed to be
disjoint, meaning that each edge total $t(e)$ in \eqref{eqw2} is
counted at most once.

\subheading{The downset non-negativity property} Several of our
algorithms, specifically for the approximate-threshold problem, take
advantage of a natural combinatorial property satisfied by edge
totals~$t(e)$.  Although $t(e)$ can be negative for a particular edge
$e$, the sum $\sum_{e \in E} t(e)$ is presumed to be non-negative,
since the parallel program should not, in total, free more memory than
it allocates.  We can generalize this property to subsets of edges,
called downsets, where a subset $S \subseteq E$ is a \defn{downset}
if, for each edge $e\in S$, every edge $e' < e$ is also in~$S$.  The
\defn{downset-non-negativity property} requires that, for every
downset $S \subseteq E$, $\sum_{e \in S} t(e) \ge 0$.  This property
corresponds to the real-world requirement that at no point during the
execution of a parallel program can the total memory allocated be net
negative.





\section{An Exact Algorithm with $O(p)$ Overhead}\label{sec:exact}

This section presents \proc{ExactOff}, an $O(|E| \cdot p)$-time
offline algorithm for exactly computing the high-water marks
$H_1(G), \ldots, H_p(G)$ of a computation DAG $G$ for all numbers of
processors $1,\ldots,p$.  We first give a simple dynamic-programming
algorithm which runs in time $O(|E| \cdot p^2)$.  We describe how
\proc{ExactOff} optimizes this simple algorithm.  We then perform an
amortization argument to prove that \proc{ExactOff} achieves a running
time of $O(|E| \cdot p)$. \secref{mem} discusses how to adapt the
algorithm in order to run in an online fashion, executing along with
the parallel program being analyzed, and introducing only $O(p)$
additional memory overhead.

The algorithm exploits the fact that $G$ can be recursively
constructed via series and parallel combinations, as
\secref{problemstatement} describes.  The algorithm builds on top of
this recursive structure.  Note that one can construct a recursive
decomposition of a series-parallel DAG $G$ in linear
time~\cite{Valdes78}. 

Given a parallel program represented by a series-parallel DAG $G$, and
a number of processors $p$, we define the $(p + 1)$-element array
$R_G = (R_G[0], R_G[1], \ldots, R_G[p])$ so that, for $i > 0$,
$R_G[i]$ is the memory high-water mark for $G$ over all antichains of
size exactly~$i$.  We define $R_G[0]$ to be $\max(0, t(G)),$ where
$t(G) := \sum_{e \in G} t(e)$.  For $i > 0$, if the
graph $G$ contains no $i$-edge antichains, then $R_G[i]$ is defined to
take the special value $\NULL$, treated as $-\infty$.

One can compute $H_p(G)$ from the array $R_G$ using the identity
$H_p(G) = \max_{i = 1}^p R_G[i]$.  Our goal is therefore to
recursively compute $R_G$ for the given DAG~$G$.

\subheading{An $O(|E| \cdot p^2)$-time algorithm} We begin with a
simple algorithm that computes $R_G$ using the recursive
series-parallel decomposition of~$G$. When $G$ consists of a single
edge $e$, we have $R_G[0] = \max(0, t(e))$, $R_G[1] = m(e)$, and
$R_G(2), \ldots, R_G[p] = \NULL$.

Suppose that $G$ is the parallel combination of two graphs $G_1$
and $G_2$. Then,
\begin{equation*}
R_G[i] = \begin{cases}
\max(0, t(G)) &\text{if $i = 0$,}\\
\max_{j = 0}^{i} R_{G_1}[j] + R_{G_2}[i - j] &\text{ otherwise}.
\end{cases}
\end{equation*}
In the second case, if either of $R_{G_1}[j]$ or $R_{G_1}[i
  - j]$ are $\NULL$, then their sum is also defined to be
$\NULL$.  Moreover, note that the definitions of $R_{G_1}[0]$ and
$R_{G_2}[0]$ ensure that suspended components are treated correctly in
the recursion.

Suppose, on the other hand, that $G$ is the series combination of two
graphs $G_1$ and $G_2$. Then $R_G$ can be expressed in terms of
$R_{G_1}$, $R_{G_2}$, $t(G_1)$, $t(G)$ using the equation,
\begin{equation*}
R_G[i] = \begin{cases}
\max(0, t(G)) &\text{if $i = 0$,}\\
\max(R_{G_1}[i], t(G_1) + R_{G_2}[i]) &\text{ otherwise}.
\end{cases}
\end{equation*}

Combining the above cases yields an $O(|E| \cdot p^2)$-time algorithm
for computing $R_G$.



\subheading{Achieving a running time of $O(|E| \cdot p)$} To optimize
the simple algorithm, we define, for a DAG $G$, the value $s(G)$ to be
the size of the largest antichain of edges in $G$, or $p$ if $G$
contains an antichain of size $p$ or larger.  The value $s(G)$ is easy
to compute recursively using the recursion
$s(G) = \min(s(G_1) + s(G_2), p)$, when $G$ is the parallel
combination of components $G_1$ and $G_2$, and
$s(G) = \max(s(G_1), s(G_2))$, when $G$ is the series combination of
$G_1$ and $G_2$.

\proc{ExactOff} optimizes the simple dynamic program as follows.
Suppose that $G$ is a parallel combination of components $G_1$ and
$G_2$. Notice that $R_{G_1}[i] = \NULL$ whenever $i > s(G_1)$ and
$R_{G_2}[i] = \NULL$ whenever $i > s(G_2)$. It follows that,
\begin{equation}
R_G[i] = \begin{cases}
\max(0, t(G)) &\text{if $i = 0$,}\\
\displaystyle\max_{\substack{0 \le j \le i, \\ j \le s(G_1), \\ (i - j) \le s(G_2)}} R_{G_1}[j] + R_{G_2}[i - j] &\text{ o.w.},
\end{cases}
\label{eqoptimizedrecursion}
\end{equation}
where the max for the second case is defined to evaluate to $\NULL$ if it
has zero terms.

\begin{theorem}
  For a series-parallel DAG $G = (V, E)$, \proc{ExactOff} recursively
  computes $R_G$ in time $O(|E| \cdot p)$.
\label{thmEp}
\end{theorem}

To prove Theorem \ref{thmEp}, let us consider the time needed to
compute $R_G$ when $G$ is obtained by combining two subgraphs $G_1$
and $G_2$ in parallel. For each value of $i \le s(G_1)$ and of $i - j
\le s(G_2)$, the term $R_{G_1}[i] + R_{G_2}[i - j]$ will appear in
\eqref{eqoptimizedrecursion} for exactly one index~$i$. It follows
that the total time to compute $R_G$ from $R_{G_1}$ and $R_{G_2}$ is
at most $O(p + s(G_1) \cdot s(G_2))$.

Since parallel combinations cost $O(p + s(G_1) \cdot s(G_2))$ and
series combinations cost $O(p)$, Theorem \ref{thmEp} reduces to,
  $$\sum_{(G_1, G_2) \in \mathcal{C}} s(G_1) \cdot s(G_2) \le O(|E|
\cdot p),$$ where the set $\mathcal{C}$ consists of all parallel
combinations in the recursive construction of~$G$.

\begin{lemma}
  Let $G = (V, E)$ be the series-parallel DAG modeling some parallel
  program's execution. Then,
  $$\sum_{(G_1, G_2) \in \mathcal{C}} s(G_1) \cdot s(G_2) \le O(|E|
  \cdot p),$$ where the set $\mathcal{C}$ consists of all parallel
  combinations in the recursive construction of $G$.
  \label{lemamortizedcosts}
\end{lemma}
\begin{proof}
  Call a parallel combination between two components $G_1$ and $G_2$
  \defn{fully-formed} if $s(G_1) = s(G_2) = p$. We claim that there
  are at most $O(|E| / p)$ fully-formed parallel combinations in the
  recursive construction of~$G$.  Consider the recursive construction
  of $G$ from edges via series and parallel combinations.  Each
  fully-formed parallel combination reduces the total number of
  components $H$ satisfying $s(H) = p$ by one. On the other hand, the
  number of components satisfying $s(H) = p$ can only be increased
  when two components $H_1, H_2$ satisfying $s(H_1), s(H_2) < p$ are
  combined to form a new component $H$ satisfying $s(H) = p$.  The
  total number of such combinations is at most $|E| / p$, since each
  such $H$ absorbs at least $p$ edges.  Since the number of components
  satisfying $s(H) = p$ is incremented at most $|E| / p$ times, it can
  also be decremented at most $|E| / p$ times, which limits the number
  of fully-formed parallel combinations to $|E| / p$.

  Using the bound on the number of fully-formed parallel combinations,
  we have that
  \begin{align*}
  \sum_{(G_1, G_2) \in \mathcal{F}} s(G_1) \cdot s(G_2)
    \le O\left(\frac{|E|}{p} \cdot p^2 \right) \le O(|E| \cdot p),
  \end{align*}
  where $\mathcal{F}$ is the set of fully-formed parallel combinations.

  To complete the proof of the lemma, it suffices to show 
  \begin{equation}
    \sum_{(G_1, G_2) \in \overline{\mathcal{F}}} s(G_1) \cdot s(G_2) \le O(|E| \cdot p),
    \label{eqnonfullamortizedcosts}
  \end{equation}
  where $\overline{\mathcal{F}}$ is the set of non-fully formed
  parallel combinations.

  We prove \eqref{eqnonfullamortizedcosts} with an amortization
  argument.  Consider the recursive construction of $G$ from edges via
  series and parallel combinations.
  Before beginning the combinations, we assign $2p - 1$ credits to
  each edge $e \in E$.  Every time two components $G_1$ and $G_2$ are
  combined in parallel and $G_1$ satisfies $s(G_1) < p$, we deduct
  $s(G_2)$ credits from each edge in~$G_1$. Similarly, if $s(G_2) <
  p$, we deduct $s(G_1)$ credits from each edge in $G_2$.  Note that
  if both $s(G_1) < p$ and $s(G_2) < p$, then we deduct credits from
  the edges in both components.

  The number of credits charged for each non-fully-formed parallel
  combination is at least $s(G_1) \cdot s(G_2)$.  Thus the total
  number of credits deducted from all edges over the course of the
  construction of $G$ is at least the left side of
  \eqref{eqnonfullamortizedcosts}. In order to prove
  \eqref{eqnonfullamortizedcosts}, it suffices to show that every edge
  still has a non-negative number of credits after the construction of
  $G$.

  Consider an edge $e \in E$ as $G$ is recursively constructed. Define
  $H_t$ to be the component containing $e$ after $t$ steps in the
  construction, $c_t$ to be the total amount of credit deducted from
  $e$ in the first $t$ steps, and $a_t$ to be the size of the largest
  antichain in $H_t$. We claim as an invariant that $c_t \le
  a_t$. Indeed, whenever $r = c_t - c_{t - 1}$ credits are deducted
  from $e$ during some step $t$, the parallel combination during that
  step also increases $a_t$ to be at least $r$ larger than
  $a_{t - 1}$.

  Since $s(H_t) = \min(a_t, p)$, the invariant tells us that whenever
  $e$ is in a component $H_t$ with $s(H_t) < p$, the total amount
  $c_t$ deducted from $e$ so far must satisfy $c_t < p$. Prior to the
  step $t$ in which $s(H_t)$ finally becomes $p$, the total amount
  deducted from $e$ is at most $p - 1$. During the step $t$ when
  $s(H_t)$ becomes $p$, at most $p$ credits can be deducted from
  $e$. And after the step $t$ when $s(H_t)$ becomes $p$, no more
  credits will ever be deducted from $e$. Thus the total deductions
  from $e$ sum to at most $2p - 1$, as desired.
\end{proof}


\section{An Online (Memory-Efficient) Algorithm}\label{sec:mem}

The \proc{ExactOff} algorithm in \secref{exact} computes $\HH_p(G)$ by
considering the construction of a computation DAG $G$ using only
series and parallel combinations. Although in principle any
series-parallel DAG can be constructed using only these combinations,
doing so in an online fashion (as the parallel program executes) can
require substantial memory overhead. In particular, parallel programs
implemented in Cilk implicitly contain a third primitive way of
combining components: multi-spawn combinations (see Figure
\ref{figmultispawns}). A multi-spawn combination corresponds with all
of the child spawns (i.e., \CilkSpawn statements) of a thread that
rejoin at a single synchronization point (i.e., at a \CilkSync).

  \begin{figure}
    \begin{center}
      \includegraphics[scale = .4]{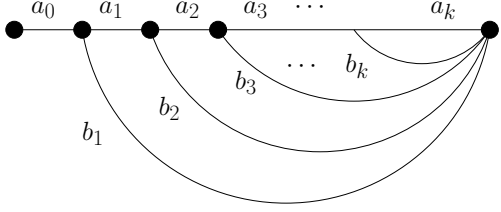}
    \end{center}

    \caption{A multi-spawn combination. The components
      $a_0, \ldots, a_k$ and $b_1, \ldots, b_k$ are combined into a
      single component. If executed on a single processor in Cilk, the
      order of execution would be
      $a_0, b_1, a_1, b_2, a_2, \ldots, b_k, a_k$.}
    \label{figmultispawns}
  \end{figure}

  When a multi-spawn combination is executed on a single processor,
  the execution traverses the components in the order $a_0, b_1, a_1,
  b_2,a_2, \ldots, b_k, a_{k}$. If one wishes to use the recursions
  from \secref{exact} in order to compute $R_G$ for a multi-spawn
  combination, then one must store the recursively computed values for
  each of $a_0, b_1, a_1, \ldots, a_k$ before any series or parallel
  combinations can be performed. After computing the values,
  one can then combine $a_{k}$ and $b_k$ in parallel, combine this
  with $a_{k - 1}$ in series, combine this with $b_{k - 1}$ in
  parallel, and so on.

  When $k$ is large, storing $\Theta(k)$ recursive values at a time
  can be impractical (though in total the multiplicative memory
  overhead of \proc{ExactOff} will still be bounded by $O(p)$ times
  the span of the parallel program). If one could instead design a
  recursion in which each multi-spawn combination could be performed
  using $O(p)$ space, then the recursive algorithm would be guaranteed
  to use no more than $O(p \cdot D)$ space, where $D$ is the maximum
  stack depth of the parallel program in Cilk.

  \ifconference
\else
  Appendix \ref{secappendixonlineexact} presents the \proc{ExactOn}
  \fi
  algorithm, which implements this alternative recursion. The
  amortized analysis in \secref{exact} fails to carry over to
  \proc{ExactOn}, because the work in the new algorithm can no longer
  be directly charged to the growth of components. Instead, we employ
  a more sophisticated amortized analysis in which components of the
  graph are assigned sets of leader vertices, and the work by the
  algorithm is charged to the leader vertices in such a way so that no
  vertex receives a charge of more than~$O(p)$.

  \section{Online Approximation in Linear Time}\label{sec:approx}

  \tbsnote{TODO: Give names to these algorithms.}

  This section considers the approximate threshold version of the
  $p$-processor memory high-water mark problem. In particular, we give
  a linear-time online algorithm \proc{ApproxOn} that processes a
  computation DAG $G=(V,E)$ and returns a boolean with the following
  guarantee: a return value of $0$ guarantees that $\HH_p(G) \le M$,
  while a return value of $1$ guarantees that $\HH_p(G) > M / 2$.

  \proc{ApproxOn} will compute the high-water mark over a special
  class of antichains that satisfy a certain property that we call
  \defn{stripped robustness}. Intuitively, the stripped robustness
  property requires that every edge $e$ in the antichain contributes a
  substantial amount (at least $M / 2p$) to the antichain's water
  mark. The algorithm solves the approximate threshold problem by
  computing the \emph{infinite-processor} water mark over all stripped
  robust antichains, and then inferring from this information about
  $\HH_p(G)$.

  \secref{strippedrobust} defines what it means for an antichain to be
  stripped robust and proves the correctness of the \proc{ApproxOn}
  algorithm.  \secref{strippedrobustcompute} describes an online
  recursive procedure for computing the quantity $h$ needed by the
  algorithm in linear time $O(|E|)$. The \proc{ApproxOn} algorithm
  uses space at most $O(D)$ where $D$ is the maximum stack depth
  during an execution of the parallel program being analyzed. A
  simpler offline algorithm is also given in \ifconference
\else Appendix
  \ref{secappendixsimpleapprox}.  \fi

  \subsection{Stripped Robust Antichains}\label{sec:strippedrobust}

  This section defines a special class of antichains that we call
  stripped robust. We prove that, by analyzing stripped robust
  antichains with arbitrarily many processors, we can deduce
  information about $\HH_p(G)$.

  An antichain $A$ is \defn{stripped robust} if it satisfies two
  requirements:
  \begin{itemize}
  \item \textbf{Large Local Contributions of Edges: }We define the
    \defn{local contribution} $L_A^\bullet(x_i)$ of each edge $x_i \in
    A$ to the water mark $W(A)$ to be the value
    $W(A) - W(A \setminus \{x_i\})$. In order for $A$ to be a stripped
    robust antichain, each $x_i$ must satisfy $L_A^\bullet(x_i) >\frac{M}{2p}$.
  \item \textbf{Large Edge Contributions of Non-Critical Components:
  }For each multi-spawn combination $a_0, b_1, a_1, \ldots, a_k$ in
    $G$, if the component $b_i$ contains at least one $x_i$, and if
    $a_i \cup b_{i + 1} \cup \cdots \cup a_k$ contains at least one
    other $x_j$, then we call $b_i$ a \defn{non-critical}
    component. Define the \defn{local contribution} of $b_i$ to be
    $L_A^\bullet(b_i) = W(A) - W(A \setminus b_i)$, the reduction in
    water mark obtained by removing from $A$ the edges also contained
    in $b_i$. In order for $A$ to be a stripped robust antichain, each
    non-critical component $b_i$ must satisfy $L_A^\bullet(b_i) > \frac{M}{2p}$.
  \end{itemize}

  The \defn{$p$-processor robust memory high-water mark}
  $\HH_p^\bullet(G)$ is defined to be
$$\HH^\bullet_p(G) = \max_{A \in \mathcal{S}, \ |A| \le p} W(x_1,
  \ldots, x_q),$$ where $\mathcal{S}$ is the set of stripped robust
  antichains in $E$.

The first step in our approximate-threshold algorithm \proc{ApproxOn}
will be to compute the infinite-processor robust memory high-water
mark $\HH_\infty^\bullet(G)$. Then, if $\HH_\infty^\bullet(G) \le M /
2$, our algorithm returns 0, and if $\HH_\infty^\bullet(G) > M / 2$,
our algorithm returns 1.

Computing $\HH_\infty^\bullet(G)$ can be done online with constant
overhead using a recursive algorithm described in Section
\ref{sec:strippedrobustcompute}. The computation is made significantly easier,
in particular, by the fact that it is permitted to consider the
infinite-processor case rather than restricting to $p$ processors or
fewer.

On the other hand, the fact that $\HH_\infty^\bullet(G)$ should tell
us anything useful about $\HH_p(G)$ is non-obvious. In the rest of
this section, we will prove the following theorem, which implies the
correctness of the \proc{ApproxOn} algorithm:

  \begin{theorem}  If $\HH_\infty^\bullet(G) \le M / 2$, then $\HH_p(G)
      \le M$, and if $\HH_\infty^\bullet(G) > M / 2$, then $\HH_p(G) >
      M / 2$.
    \label{thmalgcorrect2}
  \end{theorem}

  It turns out that Theorem \ref{thmalgcorrect2} remains true even if
  the second requirement for stripped robust antichains is removed
  (i.e. that non-critical components make large contributions). In
  fact, removing the second requirement (essentially) gives the notion
  of a \defn{robust antichain} used in
  \ifconference
\else
  Appendix \ref{secappendixsimpleapprox}
  \fi
  in designing an \emph{offline}
  algorithm for the same problem. As we shall see in
  \secref{strippedrobustcompute}, the second requirement results in
  several important structural properties of stripped robust
  antichains, making an \emph{online} algorithm
  possible. The structural properties enable a recursive
  computation of $\HH_\infty^\bullet$ to handle multi-spawn
  combinations in a memory efficient fashion.

  Our analysis begins by comparing
  $\HH_p^\bullet(G)$ to $\HH_p(G)$:
  \begin{lemma} $\HH_p^\bullet (G) \ge \HH_p(G) -
    \frac{M}{2}.$
    \label{lemremovenonrobust2}
  \end{lemma}
  \begin{proof} Consider an antichain $A_1 = (x_1, \ldots, x_q)$,
    with $q \le p$, that is not stripped-robust. We wish to construct
    a stripped robust antichain $B$ satisfying $W(B) \ge W(A) - M /
    2$.
    
    Then there must either be an edge $x_i \in A_1$ satisfying
    $L_{A_1}^\bullet(x_i) \le \frac{M}{2p}$ or a non-critical component
    $b_i$ satisfying $L_{A_1}^\bullet(b_i) \le \frac{M}{2p}$. Define an
    antichain $A_2$ obtained by removing either the single edge $x_i$
    from $A_1$ (in the case where such an $x_i$ exists) or all of the
    edges in $A_1 \cap b_i$ from $A_1$ (in the case where such a $b_i$
    exists). The antichain $A_2$ contains at least one fewer edges than does $A_1$, and satisfies
    $W(A_2) \ge W(A_1) - \frac{M}{2p}.$
    
    If $A_2$ is still not stripped-robust, then we repeat the process
    to obtain an antichain $A_3$, and so on, until we obtain a
    stripped-robust antichain $A_k$.  Because the empty antichain is
    stripped-robust, this process must succeed.
    
    Since each antichain $A_i$ in the sequence is smaller than the
    antichain $A_{i - 1}$, the total number $k$ of antichains in the
    sequence can be at most $p + 1$. On the other hand, since $W(A_i) \ge W(A_{i - 1})
    - \frac{M}{2p}$ for each $i \ge 2$, we also have that
    $$W(A_k) \ge W(A_1) - (k - 1) \cdot \frac{M}{2p} \ge  W(A_1) - \frac{M}{2},$$
    as desired.
  \end{proof}

  Corollary \ref{cor:thmpart1} proves the first part of Theorem \ref{thmalgcorrect2}:

  \begin{corollary}
    If $\HH_\infty^\bullet(G) \le M / 2$, then $\HH_p(G) \le M$.
    \label{cor:thmpart1}
  \end{corollary}

  The second half of Theorem \ref{thmalgcorrect2} is given by Lemma
\ref{lemreturn2}:
\begin{lemma} If $\HH_\infty^\bullet(G) > M / 2$, then
$\HH_p(G) > M / 2$.
\label{lemreturn2}
\end{lemma}
\begin{proof} Since
  $\HH_\infty^\bullet(G) > M / 2$, there are two cases:

  \vspace{2pt}
  \noindent \textbf{Case 1: There is a stripped robust antichain
    $A = (x_1, \ldots, x_q)$ with $q \le p$ such that $W(A) > M /
    2$. }In this case, we trivially get that $\HH_p(G) > M / 2$.

  \vspace{2pt}
  \noindent \textbf{Case 2: There is a stripped robust antichain
    $A = (x_1, \ldots, x_q)$ with $q > p$ such that $W(A) > M /
    2$. }This case is somewhat more subtle, since the large number of
  edges in the antichain $A$ could cause $W(A)$ to be much larger than
  $\HH_p(G)$. We will use the stripped robustness of $A$ in order to
  prove that the potentially much smaller antichain
  $B = (x_1, \ldots, x_p)$ still has a large water mark
  $W(B) > \frac{M}{2}$.

  Note that we cannot simply argue that $L_B^\bullet(x_i) \ge
  L_A^\bullet(x_i)$ for each $i \in \{1, \ldots, p\}$. In particular,
  the removal of edges from $A$ may significantly change the local
  contributions of the remaining edges. Nonetheless, by exploiting the
  downset-non-negativity property we will still prove that $W(B) >
  \frac{M}{2}$.

  For a given edge $x_i \in A$, define $\mathcal{T}_i$ to be the
  set of companion components $T$ to the antichain $A$ such that
  $T$ is not a companion component to $A \setminus \{x_i\}$. Define
  $P_i$ to be the set of edges $e$ such that $e < x_i$ but $e \not< x_j$
  for any other $x_j \in A$. Then the local contribution
  $L_A^\bullet(x_i)$ of $x_i$ to $A$ satisfies,
  \begin{equation}
    L_A^\bullet(x_i) \le m(x_i) + \sum_{T \in \mathcal{T}_i}
    \sum_{e \in T} t(e) + \sum_{e \in P_i} t(e).
    \label{eqlocalcontrib}
  \end{equation}
  (This would be an exact
  equality if not for the fact that removing $x_i$ from $A$ can also
  introduce \emph{new} companion components, which in turn reduces
  $L_A^\bullet(x_i)$.)

  Let $S_i$ denote the quantity on the right side of
  \eqref{eqlocalcontrib}. Since $A$ is a stripped robust antichain, $S_i \ge M / 2p$ for each $i$.

  Now let us consider the water mark $W(B)$. For each $i \in \{1,
  \ldots, p\}$, each component $T \in \mathcal{T}_i$ is a companion
  component to $B$, just as it was to~$A$.  Moreover, each edge $e
  \in P_i$ continues to contribute $t(e)$ to the water mark of
  $B$. Define $\mathcal{T}$ to be the set of companion components
  to $B$ that are not in any of $\mathcal{T}_1, \ldots,
  \mathcal{T}_p$, and $P$ to be the set of edges $e$ satisfying $e <
  x_i$ for some $i \in \{1, \ldots, p\}$ but $e \not\in P_1 \cup
  \cdots \cup P_p$. Then the water mark $W(B)$ can be written as

  \begin{align*}
    W(B) & = \sum_{i = 1}^p S_i + \sum_{T \in \mathcal{T}}
    \sum_{e \in T} t(e) + \sum_{e \in P} t(e) \\
        & \ge M/2 + \sum_{T \in \mathcal{T}}
    \sum_{e \in T} t(e) + \sum_{e \in P} t(e).
  \end{align*}

  Since each $T \in \mathcal{T}$ satisfies $\sum_{e \in T}t(e) > 0$
  (or else $T$ would not be a companion component to $B$), 
  $$W(B) \ge M / 2 + \sum_{e \in P} t(e).$$

  In order to complete the proof that $W(B) > M / 2$, it suffices by
  the downset-non-negativity property to show that $P$ is a
  downset. Notice that $P$ can be rewritten as
  \begin{align*}
    P & = \{e \in E \mid e < x_i \text{ for some } i = 1, \ldots, p\} \\
    & \cap \{e \in E \mid e < x_i \text{ and } e < x_j \text{ for some } x_i \neq x_j \in A\} \\
    & = \left(\bigcup_{i = 1}^p \{e < x_i\}\right) \cap \left(\bigcup_{x_i \neq x_j \in A} \{e < x_i\} \cap \{e < x_j\} \right).
  \end{align*}

  Since the downset property is closed under unions and intersections,
  it follows that $P$ is a downset.
  \end{proof}

\subsection{Recursively Computing $\HH_{\infty}^\bullet(G)$}\label{sec:strippedrobustcompute}

This section discusses a recursive algorithm for computing
$\HH_\infty^\bullet(G)$ in linear time and in an online fashion. This
algorithm can then be used within \proc{ApproxOn} to obtain a
linear-time online algorithm for the approximate threshold problem.
The algorithm treats $G$ as being recursively constructed via series
and multi-spawn combinations. For each multi-spawn combination, we
assume we are recursively given the computed values for $a_0, b_1,
a_1, \ldots, a_k$, one after the other.  Because $k$ may be large, the
recursion is not permitted to store these values. Instead it stores a
constant amount of metadata that is updated over the course of the
multi-spawn combination.

Finding a water-mark-maximizing stripped robust antichain $A$ in a
multi-spawn combination $C = (a_0, b_1, \ldots, a_k)$ is complicated
by the following subtlety: if we choose to include an edge in one of
the $b_j$'s, then this may reduce the local contribution of any edges
included in later $a_j$'s and $b_j$'s, resulting in those edges being
unable to be included in the antichain. Therefore, greedily
adding edges to the antichain $A$ as we recursively execute $a_0, b_1,
\ldots, a_k$ may not result in an optimal stripped robust antichain.

The second requirement for stripped robust antichains (that
non-critical components must make large edge contributions) is
carefully designed to eliminate this problem. It allows
us to prove the following lemma, which characterizes how non-critical
components behave in water-mark-maximizing antichains $A$ that
contain multiple edges.\footnote{In fact, the same lemma would be true
  if we removed the restriction that $A$ contain multiple edges. The
  restriction is necessary for our applications of the lemma,
  however.}

\begin{lemma}
  Consider a multi-spawn combination $C$ with components $a_0, b_1,
  a_1, \ldots, a_{k}$. Consider $i \in \{1, \ldots, k\}$, and suppose
  $A$ is a stripped robust antichain in $C$ that (1) contains multiple
  edges; (2) contains at least one edge in $a_i, b_{i + 1}, \ldots,
  a_{k}$; and (3) achieves the maximum water mark over all stripped robust
  antichains in $C$ that contain multiple edges.

 Let $t(b_i) = \sum_{e \in b_i} t(e)$, and let $m(b_i)$ denote the
 water mark of the best stripped robust antichain in $b_i$. (Note that
 $m(b_i)$ considers only the subgraph $b_i$.)
 \begin{itemize}
 \item If $t(b_i) > 0$ and $m(b_i) \le t(b_i) + \frac{M}{2p}$, then
   $b_i$ is a companion component of $A$.
 \item If $t(b_i) \le 0$ and $m(b_i) \le \frac{M}{2p}$, then $b_i$ is
   not a companion component of $A$ and  does not contribute
   any edges to $A$.
 \item If $m(b_i) > \max(0, t(b_i)) + \frac{M}{2p}$, then $A$ restricted
   to $b_i$ is a stripped robust antichain with water mark $m(b_i)$.
 \end{itemize}
 \label{lemcandidateantichains}
\end{lemma}
 \begin{proof}

   Any edges that $b_i$ contributes to $A$ must form a stripped robust
   antichain in $b_i$. The water mark $s$ of that antichain
   within $b_i$ can be at most $m(b_i)$. It follows that
   $$L_A^\bullet(b_i) \le m(b_i) - \max(0, t(b_i)),$$ since the removal of the
   edges in $b_i$ from $A$ will have the effect of (a) reducing the
   water mark by $s$ and (b) introducing $b_i$ as a companion
   component to $A$ if $t(b_i) > 0$.

   Since $m(b_i) - \max(0, t(b_i)) \le \frac{M}{2p}$ in the first two
   cases of the lemma, $b_i$
   cannot contribute any edges to $A$ in these cases. This ensures that in the first
   case $b_i$ will be a companion component of $A$, and in the
   second case $b_i$ will neither be a companion component nor
   contribute any edges.

   The third case of the lemma is somewhat more subtle. Suppose that
   $m(b_i) > \max(0, t(b_i)) + \frac{M}{2p}$. We wish to show that $A$
   restricted to $b_i$ is a stripped robust antichain with water mark
   $m(b_i)$. If $A$ contains at least one edge in $b_i$, then since
   $A$ has maximum water mark over multi-edge stripped robust
   antichains in $C$, it must be that $A$ restricted to $b_i$ is
   stripped robust and has water mark $m(b_i)$, as desired.

   Suppose, on the other hand that $A$ contains no edges in $b_i$. We
   will show that $A$ does not achieve the maximum water mark over all
   stripped robust antichains in $C$ that contain multiple edges.
   Define $A'$ to be $A$ with the addition of edges in $b_i$ so that
   $A'$ restricted to $b_i$ is stripped robust and has water
   mark $m(b_i)$. Since $m(b_i) \ge \max(0, t(b_i)) + \frac{M}{2p}$,
   the water mark of $A'$ must be more than $\frac{M}{2p}$ greater
   than that of $A$.

   Since $A$ has maximum water mark, and $A'$ has a larger water mark,
   $A'$ must no longer be stripped robust. Notice, however, that the
   local contribution of $b_i$ in $A'$ is greater than $\frac{M}{2p}$,
   and the local contributions of the other non-critical components of
   $C$ in $A'$ are the same as in $A$. Thus the only way that $A'$ can
   no longer be stripped robust is if there is a single edge $x_j \in
   A'$ with local contribution at most $\frac{M}{2p}$ to $W(A')$. Note
   that $x_j \not\in b_i$, and thus in order so that $x_j$'s local
   contribution to $W(A)$ can differ from its local contribution to
   $W(A')$, $x_j$ must be the only edge from $A$ that is contained in
   any of the components $a_i, b_{i + 1}, \ldots, a_{k}$. Define $A''$
   to be $A'$ with the edge $x_j$ removed. Note that $A''$ has at
   least as many edges as did $A$ initially, and is thus still a
   multi-edge antichain.

   Since the local contribution $L^\bullet_{A'}(x_j)$ of $x_j$ to $A'$
   was at most $\frac{M}{2p}$, the water mark of $A''$ still exceeds
   that of $A$. We claim, however, that $A''$ is a stripped robust
   antichain, contradicting that fact that $A$ has maximum water mark
   out of all multi-edge stripped robust antichains. If $b_i$
   contains multiple edges in $A''$, then the fact that those edges
   form a stripped robust antichain when restricted to $b_i$, and that
   the other noncritical components and edges in $A''$ have the
   same local contributions to $A''$ as they did to $A$, ensure that
   $A''$ is a stripped robust antichain. If, on the other hand, $b_i$
   contains a single edge in $A''$, then the removal of that edge
   would reduce $W(A'')$ by at least as much as would have the removal
   of the edge $x_j$ from the original antichain $A$ (since both
   removals result in the same antichain).  Thus $L_{A''}^\bullet(x_i)
   \ge L_A^\bullet(x_j) > \frac{M}{2p}$, ensuring that $A''$ is a
   stripped robust antichain. Since $A''$ has a larger water mark than
   does $A$, we have reached a contradiction, completing the proof of the
   third case of the lemma.
 \end{proof}

Our algorithm for computing $\HH^\bullet_\infty(G)$ recursively computes
 three quantities for each component $C$ of the graph. 
 \begin{itemize}
 \item The total allocation and freeing work done in $C$,
   $$\memtotal = \sum_{e \in C} t(e).$$
\item The memory high-water mark  with one processor,
  $$\maxchain = \HH_1(C).$$
 \item The infinite-processor high-water mark restricted
   to stripped robust antichains containing \emph{more} than one edge:
  $$\maxrobust = \max_{A \in \mathcal{S}, \ |A| > 1} W(A),$$ where
   $\mathcal{S}$ is the set of stripped robust antichains in $C$.  If
   $C$ contains no such multi-edge antichains, then $\maxrobust =
   \NULL$, which is treated as $-\infty$.
 \end{itemize}

 The special handling in the recursion of antichains with only one
 edge (i.e., by $\maxchain$) is necessary because the local
 contribution of that edge $x$ is not yet completely determined until
 at least one other edge is added to the antichain. On the other hand,
 once a stripped robust antichain contains multiple edges, the local
 contribution of each edge is now fixed, even if we combine this
 antichain with other antichains as we recursively construct the
 graph. This allows for all multi-edge stripped robust antichains to
 be grouped together in the variable $\maxrobust$.

As a base case, for a component $C$ consisting of a single edge
 $e$, we initialize the variables as follows: $\memtotal = t(e)$,
 $\maxchain = m(e)$, and $\maxrobust = \NULL$. 

 When we combine two components $C_1$ and $C_2$ in series to build
 a new component $C$, we have,
   $$C.\memtotal = C_1.\memtotal + C_2.\memtotal,$$
 \begin{multline*}
 C.\maxchain = \\
 \max(C_1.\maxchain, C_1.\memtotal_{C_1} + C_2.\maxchain),
 \end{multline*}
 \begin{multline*}
   C.\maxrobust = \\ \max(C_1.\maxrobust, C_1.\memtotal +
   C_2.\maxrobust).
 \end{multline*}

 In the computations of $C.\maxchain$ and $C.\maxrobust$ we implicitly
 use the fact that every antichain in $C$ must either be in $C_1$ or
 in $C_2$. Moreover, the antichains in $C_2$ have water mark
 $C_1.\memtotal$ greater in $C$ than they did in $C_2$.

 Note that the computation of $C.\maxrobust$ would not be correct if
 $\maxrobust$ were also considering single-edge antichains. In
 particular, the local contribution of an edge in a single-edge
 antichain $A$ in $C_2$ will differ from the local contribution of the
 same edge in the same antichain in $C_1 \cup C_2$, allowing it to
 possibly form a stripped robust antichain in one but not the
 other. Because $\maxrobust$ considers only multi-edge antichains,
 however, this is not a problem.

The recursion for combining components in a multi-spawn combination
is more sophisticated. Consider a multi-spawn combination $C$ as in
Figure \ref{figmultispawns} with components $C_1 = a_0, C_2 = b_1,
C_3 = a_1, \ldots, C_{2k + 1} = a_k$. As our algorithm receives
information on each of $C_1, C_2, C_3, \ldots$, it must gradually
construct the best multi-edge stripped robust antichain in the
multi-spawn component. Lemma \ref{lemcandidateantichains} ensures that
this is possible, because the role that each $b_i$ plays in such an
antichain depends only on whether any additional edges will be
included from later components $a_i, b_{i + 1}, \ldots$, and not on
the specific properties of the components.

Nonetheless, the bookkeeping for the recursion is made subtle by the
handling of suspended components and other casework. We defer the
full recursion to
\ifconference
\else
Appendix \ref{secmultispawnrecursion}.
\fi

\secput{eval}{Empirical evaluation}

This section discusses the implementation and evaluation of Cilkmem on a
suite of benchmark programs as well as on a large image processing pipeline
performing image alignment.

%
%

\subsection{Implementation}\label{sec:implementation}

We implemented Cilkmem as a CSI tool~\cite{SchardlDeDo17} written in C++ for
the Tapir compiler~\cite{SchardlMoLe17}. The following discussion describes how
these facilities are used to implement Cilkmem's algorithms for MHWM analysis.  

The Cilkmem CSI tool tracks the evolution of a program's series-parallel DAG by
inserting shadow computation before and after the instructions used by the
Tapir compiler to represent fork-join parallelism. The language constructs used
by the program-under-test to represent fork-join parallelism (e.g., Cilk's
\code{cilk_spawn} and \code{cilk_sync} keywords) are lowered to Tapir's
\textit{detach} and \textit{sync} instructions during compilation.  The CSI
framework provides instrumentation \emph{hooks} that enable a tool to insert
\emph{shadow computation} before and after instructions in the compiler's
intermediate representation of the program.

\begin{table*}[t]
\begin{tabular*}{\textwidth}{l@{\extracolsep{\fill}}lrrrrrrrrr}
 & & \multicolumn{7}{c}{$T_{\text{\textit{exact}}}/{T_1}$} & \\ \cmidrule[0.3pt]{3-10}  
\textit{Benchmark\enspace \enspace}& \textit{Input size} & \multicolumn{1}{c}{$\scriptstyle{p=32}$} & \multicolumn{1}{c}{$\scriptstyle{p=64}$} & \multicolumn{1}{c}{$\scriptstyle{p=128}$} & \multicolumn{1}{c}{$\scriptstyle{p=256}$} & \multicolumn{1}{c}{$\scriptstyle{p=512}$} & \multicolumn{1}{c}{$\scriptstyle{p=1024}$} & \multicolumn{1}{c}{$\scriptstyle{p=2048}$} & \multicolumn{1}{c}{$\scriptstyle{p=4096}$} & \multicolumn{1}{l}{${T_{\text{\textit{approx}}}}/{T_1}$} \\ \hline
strassen &  4096 x 4096 matrix
 & $1.00$ & $1.00$ & $1.00$ & $0.98$ & $1.02$ & $0.98$ & $0.98$ & $1.00$ & $0.99$ \\
nBody &  1,000,000 points
 & $1.00$ & $1.00$ & $1.00$ & $0.99$ & $0.99$ & $0.99$ & $1.00$ & $1.00$ & $1.00$ \\
lu &  4096 x 4096 matrix
 & $1.01$ & $1.02$ & $1.03$ & $1.02$ & $1.06$ & $1.61$ & $3.02$ & $5.77$ & $1.01$ \\
remDups &  100,000,000 integers & $1.03$ & $1.03$ & $1.03$ & $1.01$ & $1.02$ & $1.04$ & $1.58$ & $3.05$ & $1.02$ \\
dict &  100,000,000 integers
 & $1.11$ & $1.10$ & $1.12$ & $1.11$ & $1.10$ & $1.12$ & $1.66$ & $3.17$ & $1.10$ \\
ray &  small
 & $1.11$ & $1.11$ & $1.11$ & $1.09$ & $1.11$ & $1.21$ & $1.38$ & $2.74$ & $1.10$ \\
delaunay &  5,000,000 points
 & $1.16$ & $1.16$ & $1.17$ & $1.13$ & $1.13$ & $1.15$ & $1.16$ & $1.19$ & $1.16$ \\
nqueens &  13 x 13 board
 & $1.44$ & $1.46$ & $1.54$ & $1.66$ & $4.67$ & $7.80$ & $14.82$ & $28.53$ & $1.23$ \\
qsort &  50,000,000 elements
 & $3.63$ & $3.73$ & $4.07$ & $5.18$ & $12.75$ & $28.60$ & $53.89$ & $103.63$ & $2.61$ \\
cholesky &  2000 x 2000 matrix
 & $6.62$ & $6.86$ & $7.51$ & $10.08$ & $34.68$ & $59.23$ & $113.79$ & $219.18$ & $4.68$ \\
\hline
\end{tabular*}

\vspace{1ex}
\caption{Application benchmarks from the Cilkbench suite showing the 
overhead of Cilkmem over a single-threaded execution. The overhead is 
computed as the geometric-mean ratio of at least 5 runs with Cilkmem 
enabled and at least 5 runs of the un-instrumented program for each 
benchmark in the table.}
\vspace{-3ex}
\label{tab:cilkbench}
\end{table*}

Memory allocations and frees are tracked by Cilkmem using process-wide hooks
that intercept calls to the major allocation facilities provided by
\code{glibc} via library interpositioning~\cite[Ch. 7.13]{BryantOha15}. These
allocation facilities include \code{malloc}, \code{aligned_alloc},
\code{realloc} and \code{free}. While Cilkmem could use CSI's instrumentation
hooks to track memory allocations, the use of interpositioning allows Cilkmem
to capture calls to allocation functions performed in shared dynamic libraries
that may not have been compiled with instrumentation enabled.  Furthermore,
interpositioning makes it possible for Cilkmem to track the requested sizes of
allocations without the maintenance of an additional auxillary data structure
by prepending to each allocation a small \emph{payload} containing the size of
the allocated block of memory. This payload is retrieved at deallocation time
to determine how much memory has been freed. As an alternative to the
payload-based technique, Cilkmem can also be instructed to retrieve the size of
allocations using the Linux-specific function \code{malloc_usable_size}. The
difference between the two methods comes down to whether the allocation size
seen by Cilkmem is the \emph{requested} size or the \emph{usable} size
determined by the memory allocator, which is allowed to reserve more memory
than requested by the user. Special care is taken to ensure that the memory
activity of Cilkmem's instrumentation is properly distinguished from activity
originating from the program-under-test. 

Cilkmem separates its instrumentation and analysis logic into two separate
threads that communicate in a producer--consumer pattern. The producer thread
executes the program-under-test and generates records that keep track of the
allocation or deallocation of memory as well as the evolution of the
series-parallel structure of the program execution. These records are sent to
the consumer thread which executes either the exact or approximate MHWM
algorithm in a manner that is fidelitous with the descriptions of the online
MHWM algorithms in \secref{mem} and \secref{approx}. In order to maintain the
theoretical space-bounds of the online MHWM algorithms, the Cilkmem tool's
producer thread blocks in the rare case there is a backlog of unconsumed
records. Although it would be possible to implement the online MHWM algorithms
without the use of this producer--consumer pattern, such an approach can result
in increased instrumentation overhead due to, among other things, an increase
in instruction cache misses. 

In addition to computing the memory high-water mark, Cilkmem supports
a \emph{verbose} mode which provides the user with actionable
information for identifying the root cause of the memory high-water
mark. In verbose mode, Cilkmem constructs the full computation DAG,
annotates each edge (strand) with information about how much memory is
allocated within that strand, and outputs the resulting graph in a
graphical format. Furthermore, Cilkmem identifies the lines of code
responsible for allocations that contribute to the memory high-water
mark for a given $p$, and reports them to the user along with how many
bytes each line is responsible for.

Figure \ref{fig:example_output} shows the output genenerated by 
Cilkmem's verbose mode when run on the \emph{strassen} benchmark from the Cilkbench suite with a 
$4096\times 4096$ matrix as the input. The program's memory high-water mark 
increases by about $122$MB for each additional processor, and Cilkmem identifies line $517$ of 
source file \texttt{strassen.c} as the responsible for such increase. 
Cilkmem also reveals that the allocations performed by two other 
lines of code do not increase in size as $p$ increases.


\begin{figure}[t]
\begin{Verbatim}[fontsize=\small]
Memory high-water mark for p = 1 : 894727307 bytes
Source map for p = 1:
  [strassen.c:517]: 492013632 bytes
  [strassen.c:776]: 402653184 bytes
  [strassen.c:459]: 60491 bytes
Memory high-water mark for p = 2 : 1017641835 bytes
Source map for p = 2:
  [strassen.c:517]: 614928160 bytes
  [strassen.c:776]: 402653184 bytes
  [strassen.c:459]: 60491 bytes
Memory high-water mark for p = 3 : 1140556363 bytes
Source map for p = 3:
  [strassen.c:517]: 737842688 bytes
  [strassen.c:776]: 402653184 bytes
  [strassen.c:459]: 60491 bytes
\end{Verbatim}
\vspace{-2ex}
\caption{An example of Cilkmem's verbose-mode output for the 
\emph{strassen} benchmark. This excerpt shows the reported high-water mark 
in bytes for $p$ from $1$ to $3$ and which lines of code contribute to it.}
  \label{fig:example_output}
\vspace{-3ex}
\end{figure}

\subsection{Benchmarks}\label{sec:benchmarks}
We tested the runtime overhead of Cilkmem on ten Cilk programs from the
Cilkbench suite\footnote{The Cilkbench suite is at
\url{https://github.com/neboat/cilkbench}}. The Cilkbench suite contains a
variety of programs that implement different kind of algorithms such as
Cholesky decomposition, matrix multiplication, integer sorting, and Delaunay
triangulation.

Table \ref{tab:cilkbench} shows the geometric-mean overhead as the
ratio between an execution of the benchmark program through Cilkmem
and a serial execution of the uninstrumented program ($T_1$).  The
performance of Cilkmem was tested for both the exact algorithm and for
the approximate-threshold algorithm. When the exact algorithm was
used, Cilkmem was run with $p$ set to all powers of $2$ from $32$ to
$4096$ inclusive and its runtime is reported in the table as
$T_{\text{\textit{exact}}}$. The approximate-threshold algorithms's
runtime ($T_{\text{\textit{approx}}}$) is independent of~$p$.

As can be evinced from the results, Cilkmem's overheads are generally low and
typically result in less than a $20\%$ overhead relative to an uninstrumented
execution. Cilkmem's overheads are especially low for benchmarks that do not
exhibit substantial fine-grained parallelism. For programs that do
(e.g., \code{qsort} or \code{cholesky}), however, the exact MHWM algorithm can
incur a significant performance degradation for large values of $p$. In these
cases, the approximate-threshold algorithm is \emph{substantially} faster than
the exact algorithm.

\begin{figure}
  \begin{minipage}[t]{1.0\columnwidth}
  \begin{Verbatim}[fontsize=\small]
Differential MHWM for p = 3 -> p = 4
  [sift.cpp:355]:    63.79 MB
  [sift.cpp:259]:    21.26 MB
  [sift.cpp:319]:     7.08 MB
  [tile.cpp:2039]:    5.31 MB
  [sift.cpp:327]:   226.76 MB
  \end{Verbatim} 
  \end{minipage}
\vspace{-2ex}
\caption{Differential MHWM report on image alignment application.
Illustrates the output of Cilkmem's differential MHWM report showing the relative increase in the MHWM when increasing the number of processors from $p=3$ to $p=4$.}
\vspace{-2ex}
\figlabel{alignment-diff-report}
\end{figure}

\subsection{Optimizations}

The Cilkmem tool implements two critical optimizations to achieve low
instrumentation overheads.


Many Cilk programs do not exhibit memory activity in every strand. The
MHWM algorithms were optimized to avoid performing unecessary work
whenever a component is guaranteed not to contribute to the water mark. This
often allows Cilkmem to quickly skip over large sections of the series-parallel
structure.

As outlined in section \ref{sec:implementation}, Cilkmem utilizes two threads
which act as a producer and as a consumer. Since the data produced (allocated)
by the first thread is consumed (freed) by the second in FIFO order, the memory
management of Cilkmem's internal data structures can be greatly simplified to
avoid a large number of small allocations and deallocations. Memory is managed
in a memory pool that takes advantage of the FIFO nature of the
producer-consumer relation.

\subsection{Case Study: Multicore Image Processing Pipeline}\label{sec:casestudy} 


We conducted a case study on an existing image processing
pipeline~\cite{KalerWhWo19} that performs alignment and reconstruction of
high-resolution images produced via electron microscopy\footnote{For the
purposes of this study, we ran the pipeline of \cite{KalerWhWo19} on
full-resolution images without using FSJ.}. The alignment code processes
thousands of $8.5$ MB image tiles in order to stitch them together to form a
2D mosaic. The pipeline was designed to carefully manage memory resources
so as to be able to 2D align very large mosaics on a single multicore. The
memory usage of the application scaled predictably as a function of the size of
the mosaic, but there was an unexplained increase in the memory usage when
adding additional processors. Given the size of the individual images being
aligned, a natural expectation would be for the MHWM to increase by
approximately $8.5$ MB per processor. Empirically, however, the application's
$18$-core execution used several gigabytes more memory than the $1$-core
execution, and the precise amount of extra memory used varied from run-to-run.


We used Cilkmem to analyze the pipeline's $p$-processor MHWM\@.  Cilkmem revealed that the MHWM increased by
approximately $350$ MB per processor.  To identify the source of this
per-processor memory requirement, we used Cilkmem to generate a differential
MHWM report\footnote{Cilkmem generates differential MHWM reports via a
lightweight script that parses Cilkmem's verbose-mode output.} that attributes
an increase in the MHWM between $p$ and $p+1$ processors to particular source-code
locations. 

\figref{alignment-diff-report} shows the differential MHWM report
generated by Cilkmem for the alignment code, which reveals that lines
259, 327, and 355 of \texttt{sift.cpp} are responsible for increasing
the MHWM of the application by a total of approximately $311$ MB per
processor. These lines perform allocations to store a
difference-of-gaussian pyramid of images as a part of the
scale-invariant feature transform employed by the alignment code. This
procedure doubles the resolution of the images (to allow for subpixel
localization of features), performs $6$ GaussianBlurs on the images,
downsamples the blurred image by a factor of $2$ in each dimension,
and repeats until the image size falls below a threshold. This routine
is called on overlapping regions of image tiles. Although these
overlapping regions are typically small, the largest $1\%$ are as
large as $2$ MB in size. Whereas the original input images represent
each pixel with a single byte, this procedure generates intermediate
results that use $4$-byte floating-point values for each pixel. A
back-of-the-envelope calculation revealed that this procedure can
require $200$--$400$MB of space for tile pairs with large overlapping
regions, which conforms with Cilkmem's analysis.

This study illustrates a case in which Cilkmem allowed application developers
to make precise their memory-use projections for their pipeline, and
illuminated a source of memory blow-up when running the pipeline on very-large
multicores. Cilkmem's low instrumentation overhead made it practical to perform
frequent and iterative tests on multiple versions of the pipeline. In fact,
both the exact and approximate Cilkmem MHWM algorithms had less than $5\%$
overhead on the alignment application for $p=128$. These low overheads, coupled
with the illuminating insights provided by Cilkmem, led to the incorporation of
Cilkmem into the regression testing process for the alignment pipeline.

%
%



\secput{relatedwork}{Related Work}
 
This section overviews related work in analysis of parallel programs,
focusing on analysis of memory consumption and parallel-program
analyses that do not depend on the particulars of the task-parallel
program's scheduling.

\subheading{Related theoretical work}
From a theoretical perspective, the memory high-water mark problem
(MHWM) is closely related to the \defn{poset chain optimization
  problem (PCOP)} \cite{ShumTr96, CameronEd79, Shum90,
  MarchalNaSi18}.  In an instance of PCOP, one is given a parameter $p$
and an arbitrary DAG $G = (V, E)$ in which each edge has been assigned
a non-negative weight, and one wishes to
determine the weight of the heaviest antichain containing $p$ or fewer
edges (where the weight of each antichain is the sum of the weights of
its edges).  Shum and Trotter \cite{ShumTr96} showed that PCOP
can be solved in (substantial) polynomial time in the size of $G$
using a linear-programming based algorithm; for the special case of $p
= \infty$, an algorithm based on maximum-flows is also known
\cite{CameronEd79, Shum90, MarchalNaSi18}.

The relationship between PCOP and MHWM was previously made explicit by
Marchal \textit{et al.}\ \cite{MarchalNaSi18}, who applied algorithms
for PCOP in order to design polynomial-time algorithms for computing
the memory high-water mark of parallel algorithms.  Because none of
the known algorithms for PCOP run in (even close) to linear time,
however, the memory high-water mark algorithms of \cite{MarchalNaSi18}
are too inefficient to be used in practice.  Additionally, in order to
apply PCOP algorithms to the computation of memory high-water marks,
\cite{MarchalNaSi18} were forced to make a number of simplifying
assumptions about the parallel programs being analyzed, and their
algorithms require that the parallel programs be in what they call the
simple-data-flow model.

The difficulty of solving PCOP efficiently on arbitrary DAGs motivates
the focus in this paper on the important special case in which the DAG
$G$ is series-parallel.  Moreover, by modeling memory with arbitrary
allocations and frees (rather than using the simple-data-flow model)
we ensure that our algorithms have theoretical guarantees when applied
to analyzing arbitrary task-parallel programs.

In addition to considering the high-water mark problem, Marchal
\textit{et al.}\ \cite{MarchalNaSi18} considered the problem of
adding new dependencies to a parallel program's DAG in order to reduce
the high-water mark. They prove that this problem is
NP-complete, and empirically evaluate several
heuristics.  These techniques would be difficult to apply to most
real-world parallel programs, however, since they require the offline
analysis of the parallel program's computation DAG\@.

\subheading{Related tools}
In practice, many tools exist to measure and report on the maximum
memory consumption of a running program.  For example, the Linux
kernel tracks memory-usage information for every running
process and publishes that information through the \code{proc}
pseudo-filesystem \cite{QuinlanKe17}, including the virtual memory and
\defn{resident set size} (RSS), which is the portion of the process's
memory stored in main memory and, therefore, is upper-bounded by the
memory high-water mark.  Performance-analysis toolkits including Intel
VTune Amplifier \cite{VTune19}, the Sun Performance Analyzer
\cite{Oracle10}, the Massif tool \cite{Massif18} in the Valgrind tool
suite \cite{NethercoteSe07}, and Linux's \code{memusage} tool
\cite{SchifferKe14} measure the memory consumption of an execution of
a specified program and reports its peak stack- and heap-memory
consumption.  Like Cilkmem, these tools use intercept system calls to
dynamic memory-allocation functions, such as \code{malloc} and
\code{free}.  Unlike Cilkmem, however, all of these memory-analysis
tools gather information that is specific to how the program was
scheduled for a particular run.  These analyses do not 
analyze the worst-case memory consumption of any parallel execution of
the program on a given processor count.

Several other dynamic-analysis tools for task-parallel programs have
been developed whose analyses do not depend on the scheduling of the
program.  Tools such as Cilkview \cite{HeLeLe10} and
Cilkprof \cite{SchardlKuLe15} analyze the execution of a Cilk program
and report on the program's parallel scalability, which reflects how
much speedup the program might achieve using different numbers of
parallel processors.  Several other tools analyze
parallel memory accesses in a task-parallel program that might exhibit
nondeterministic behavior between runs of the program~\cite{FengLe99,
BenderFiGi04, UtterbackAgFi16, RamanZhSa10, RamanZhSa12,
DimitrovVeSa15}.  Like Cilkmem, the analyses performed by these tools
do not depend on how the program was scheduled for a particular run,
and instead provide insight into the behavior or performance of all
parallel executions of the program.  Unlike Cilkmem, however, these
tools do not analyze memory consumption.




\secput{concl}{Conclusion}

This paper introduces Cilkmem, a tool that analyzes the
$p$-processor memory high-water mark of fork-join
programs. Cilkmem is built on top of novel algorithms which provide
Cilkmem with both accuracy and running-time guarantees. We conclude with several directions of future work.

Although Cilkmem analyzes the behavior of parallel programs, currently
Cilkmem is forced to run these programs in serial while performing the
analysis. Extending Cilkmem to run in parallel is an
important direction of future work.

Theoretically, all of our algorithms could be implemented in parallel
at the cost of requiring additional memory. In particular, both online
algorithms adapt to this setting using total space $O(p \cdot
T_\infty)$ for the exact algorithm and $O(T_\infty)$ for the
approximate-threshold algorithm, where $T_\infty$ is the span of the
parallel program being analyzed. This can be significant, especially
for parallel programs with large multi-spawn
combinations. Furthermore, there are technical challenges in
parallelizing Cilkmem. The current instrumentation approach is not
easily amenable to parallelization since thread scheduling is hidden
by the Cilk runtime system. Finally, capturing memory allocations in a
multi-threaded program is made more difficult by the fact that each
allocation needs to be properly attributed to the correct thread and
strand.

A theoretically interesting direction of work would be to extend our
work on approximation algorithms to consider the memory-high water
mark on parallel programs with \emph{arbitrary} DAGs. Whereas
computing the exact memory-high water mark of an arbitrary DAG is
known to be difficult to do with low overhead, much less theoretical
work has been done on the approximation version of the same question.



%
\bibliographystyle{abbrv}
\bibliography{allpapers}

\ifconference
\else
\appendix

\section{Online Exact Computation of $\HH_p(G)$}\label{secappendixonlineexact}

This section describes \proc{ExactOn}, an online algorithm to compute
the exact memory high-water mark for processor counts $1,\ldots,p$.
\proc{ExactOn} adapts the $O(|E| \cdot p)$-time exact algorithm,
\proc{ExactOff}, from \secref{exact} to space-efficiently handle
multi-spawn combinations.

Formally, \proc{ExactOn} recursively computes three quantities for
each component $C$: (1) the $(p + 1)$-element array $R_C$; (2) the
total memory allocated $t(C)$ over the edges in $C$; and (3) the value
of $s(C)$. Since $s(C)$ can be recovered in time $O(p)$ from $R_C$,
the final of these quantities can be computed non-recursively for each
component.

Now consider a multi-spawn combination $C = (a_0, b_1, a_1, \ldots,
b_k, a_k)$, and let $x_1, x_2, \ldots, x_{2k + 1}$ be a consecutive
labeling of $a_0, b_1, a_2, \ldots, b_k$ (i.e., $x_1 = a_0, x_2 = b_1,
\ldots$). During a multi-spawn combination, we are given the values of
$R_{x_i}, t(x_i), s(x_i)$ for each $i = 1, \ldots, 2k + 1$, one after
another, and we wish to compute $R_{C}$ and $t(C)$ (after which we can
obtain $s(C)$ from $R_C$ in time $O(p)$).

The quantity $t(C)$ is easy to recursively compute, since $t(C) =
\sum_i t(x_i)$. What's more difficult is to obtain the array $R_C$. To
do this, as we receive $R_{x_l}, t(x_l), s(x_l)$ for each $l$, we
maintain three intermediate variables.

In order to define our intermediate variables, we must first introduce
the notions of suspended-end and ignored-end water marks of antichains
in a multi-spawn component $C$. If an antichain $A$ in $C$ contains
only edges in $b_1, \ldots, b_{k}$, and $b_t$ is the largest $t$ such
that $b_t$ contains an edge in $A$, then we say $A$ has a
\defn{suspended end} if the components $a_{t + 1}, b_{t + 2}, \ldots,
a_{k}$ form a companion component of $A$ (which occurs if the sum of
their edge costs is net positive). We cannot know whether $A$ will
have a suspended end until we have seen all of $a_{t + 1}, b_{t + 2},
\ldots, a_k$ (i.e., until we have completed the entire multi-spawn
component). Thus, for each antichain $A$ in $C$ that contains only
edges in $b_1, \ldots, b_k$, we maintain both a ``suspended-end''
version of its water mark and a ``ignored-end'' version of its water
mark. The \defn{suspended-end water mark} of $A$ is $W(A)$ if $A$ has
a suspended end, and is $W(A) + \sum_{e \in a_{t + 1}, b_{t + 2},
  \ldots, a_{k}} t(e)$ if $A$ does not have a suspended end (i.e., it
is the water mark $A$ would have if it had a suspended
end). Similarly, the \defn{ignored-end water mark} of $A$ is $W(A)$ if
$A$ does not have a suspended end, and is $W(A) - \sum_{e \in a_{t +
    1}, b_{t + 2}, \ldots, a_{k}} t(e)$ if $A$ does have a suspended
end. Additionally, for any antichain $A$ that contains an edge in some
$a_t$, we define the \defn{ignored-end water mark} of $A$ to be the
water mark of $A$; thus the ignored-end water mark is defined for all
antichains in $A$ of $C$.

As we receive $R_{x_l}, t(x_l), s(x_l)$ for each $l$, we maintain
three intermediate variables, each of which is an $O(p)$-element array
indexed either by $i = 0, \ldots, p$ or $i = 1, \ldots, p$. Two of the
three arrays, $\suspended_l$ and $\ignored_l$ are devoted to keeping
track of the largest suspended-end and ignored-end costs of antichains
seen so far. The third array, $\partialchain_l$, keeps track of the
costs of ``partially complete'' antichains, assuming that additional
edges will be added to these antichains from later $x_l$'s.

Define $l_1$ to be the index of the largest-indexed $a_i$ before or at
$x_l$, and $l_2$ to be the index of the largest-indexed $b_i$ before
or at $x_l$. After receiving $R_{x_l}, t(x_l), s(x_l)$, the $l$-th
entry of each of our three intermediate variables are updated to be
defined as follows:
\begin{itemize}
\item $\suspended_l[i]$ for $i = 1, \ldots, p$: This is the maximum
  suspended-end cost of any antichain $A$ in $x_1 \cup \cdots \cup
  x_l$ such that (1) $A$ contains exactly $i$ edges; (2) all of $A$'s
  edges are in $b_1, \ldots, b_{l_2}$. Note that here $x_1 \cup \cdots
  \cup x_l$ is treated as a multi-spawn component and the costs of the
  antichains are considered just within the graph containing $x_1,
  \ldots, x_l$ (which matters because we are considering the
  suspended-end cost of the antichain).
\item $\ignored_l[i]$ for $i = 1, \ldots, p$: This is the maximum
  ignored-end cost of any antichain $A$ in $x_1 \cup \cdots \cup x_l$
  that contains exactly $i$ edges.  If no such $A$ exists, this is
  $\NULL$, and is treated as $-\infty$.
\item $\partialchain_l[i]$ for $i = 0, \ldots, p$: Consider antichains
  $A_1, \ldots, A_{l_2}$ in $b_1, \ldots, b_{l_2}$, respectively, such
  that the total number of edges in the antichains is $i$. Then
  $\partialchain_l[i]$ is the sum of two quantities: (1) the sum of
  the edge-totals in the $a_i$'s seen so far, given by $\sum_{i =
    0}^{l_1} t(a_i)$; and (2) the maximum possible value of the sum
  $\sum_{j = 1}^{l_2} W(A_i)$, where the water mark of each $A_i$ is
  considered within only the graph $b_i$, and the water mark of an
  antichain $A_i$ containing no edges is set to $\max(0, t(b_i))$. (If
  no such $A_1, \ldots, A_{l_2}$ exist then $\partialchain_l[i] =
  \NULL$.)

  One can express $\partialchain_l[i]$ as
  \begin{equation} \partialchain_l[i] = \sum_{i = 0}^{l_1} t(a_i) + \max_{t_1 + t_2 +
  \cdots + t_{l_2} = i} \sum_{j = 1}^{l_2} R_{b_j}[t_j].
    \label{eqpartials}
  \end{equation}
\end{itemize}

We now describe how to compute the intermediate variables recursively
within a multi-spawn combination. In particular, given the
intermediate variables for $l - 1$, and given $R_{x_l}$ and $t(x_l)$,
we show how to recover the intermediate variables for $l$.

When $l = 0$, before starting the computation, we initialize each
entry of each of the intermediate variables to $\NULL$; the exception
to this is $\partialchain_0[0]$ which we initialize to zero.

Upon receiving a given $x_l$ for an odd $l$ (meaning $x_l = a_{l_1}$),
we can compute the new intermediate variables as follows:
\begin{itemize}
\item \textbf{$\suspended_l[i]$ for $i = 1, \ldots, p$: }This equals $\suspended_{l - 1}[i] + t(x_l)$.
\item \textbf{$\ignored_l[i]$ for $i = 1, \ldots, p$: }This equals
  $$\max\left(\ignored_{l - 1}[i], \max_{j = 1}^{i} \partialchain_{l - 1}[i - j]  + R_{x_l}[j]\right).$$

  In particular, the right side of the maximum considers as an option
  for $\ignored_l[i]$ the possibility that our antichain $A$ has some
  non-zero number $j$ of edges in $a_{l_1}$, and then $i - j$ edges spread
  across $b_1, \ldots, b_{l_2}$.
\item \textbf{$\partialchain_l[i]$ for $i = 0, \ldots, p$: }This is $\partialchain_{l - 1}[i] + t(x_l)$.
\end{itemize}

Upon receiving a given $x_l$ for an even $l$ (meaning $x_l =
b_{l_2}$), we can compute the new intermediate variables as follows:

\begin{itemize}
\item \textbf{$\suspended_l[i]$ for $i = 1, \ldots, p$: }This equals
      \begin{multline*}
      \max\Big(\suspended_{l - 1}[i] + t(b_{l_2}), \\
      \max_{j = 1}^{i} \partialchain_{l - 1}[i - j]  + R_{x_l}[j]\Big).
      \end{multline*}

  In particular, the right side of the maximum considers as an option
  for $\suspended_l[i]$ the possibility that our antichain $A$ has
  some non-zero number $j$ of edges in $b_{l_2}$, and then $i - j$
  edges spread across $b_1, \ldots, b_{l_2 - 1}$. Note that the
  suspended-end cost and ignored-end cost of such an antichain in $x_1
  \cup \cdots \cup x_l$ will be equal (since there is no end to be
  suspended); consequently, we will use a similar maximum to compute
  the new value of $\ignored_l[i]$.
\item \textbf{$\ignored_l[i]$ for $i = 1, \ldots, p$: }This equals
  $$\max\left(\ignored_{l - 1}[i], \max_{j = 1}^{i} \partialchain_{l - 1}[i - j]  + R_{x_l}[j]\right).$$

  As before, the right side of the maximum considers as an option
  for $\ignored_l[i]$ the possibility that our antichain $A$ has
  some non-zero number $j$ of edges in $b_{l_2}$, and then $i - j$ edges
  spread across $b_1, \ldots, b_{l_2 - 1}$.

\item \textbf{$\partialchain_l[i]$ for $i = 0, \ldots, p$: }This equals
  $$\max\left(\partialchain_{l - 1}[i] + R_{x_l}[0], \max_{j = 1}^{i} \partialchain_{l - 1}[i - j]  + R_{x_l}[j]\right),$$
  where the right side of the maximum is null for $i = 0$.

  Once again, the right side of the maximum considers as an option for
  $\ignored_l[i]$ the possibility that our antichain $A$ has some
  non-zero number $j$ of edges in $b_{l_2}$, and then $i - j$ edges spread
  across $b_1, \ldots, b_{l_2 - 1}$. The left side, on the other hand,
  represents the case where no antichain edges appear in $b_{l_2}$.
\end{itemize}

Using the recursions above, we can compute the intermediate values for
each $l$ in time $O(p^2)$. We can then compute $R_C$ for the full
multi-spawn combination $C$ using the identity
$$R_C[i] = \max(\suspended_{2k + 1}[i], \ignored_{2k + 1}[i]),$$
for $i > 0$ and $R_C[0] = \max(0, t(C))$.

For a series-parallel graph $G$, we now have a online
(space-efficient) algorithm for computing $R_G$ in time $O(|E| \cdot
p^2)$. Using a similar optimization as in the previous section, we can
improve this running time to $O(|E| \cdot p)$. In particular, for each
$\partialchain_l$, we keep track of the largest index containing a
non-null entry; when computing each of the maximums in our updates, we
can then ignore any terms involving a null entry of either
$\partialchain_{l - 1}$ or of $R_{x_l}$. This ensures that computing
the intermediate variables for a given value of $l$ takes time at most
$O(p + s_1 s_2)$, where $s_1$ is the maximum size $\le p$ of any
antichain in $b_1 \cup \cdots \cup b_{\lfloor (l - 1) / 2 \rfloor}$ and
  $s_2$ is $s(x_l)$. In particular, this means that the time for each
  $l$ is,
\begin{equation}
  O\left(p + \min\left(\sum_{j = 1}^{\lfloor (l - 1)/2 \rfloor} s(b_j), p\right) \cdot s(x_l)\right).
  \label{eqtimesubcombination}
\end{equation}

We define \proc{ExactOn} to be the online algorithm for computing
$R_{G}$ and $t(G)$ for a graph $G$, using the same recursion as
\proc{ExactOff} when combining components in series, and using the
recursion described above within each multi-spawn combination.  Using
an amortized analysis we will prove that the algorithm has
running time $O(|E| \cdot p)$.

\begin{theorem}
For a graph $G = (V, E)$ recursively constructed with series and
multi-spawn combinations, \proc{ExactOn} recursively computes
$R_G$ in time $O(|E| \cdot p)$ and using space $O(p)$ for each
combination.
\label{thmEpmem}
\end{theorem}

The fact that the product in \eqref{eqtimesubcombination} involves a
summation $\sum_{j = 1}^{\lfloor (l - 1) / 2\rfloor}s(b_j)$ means that
the simple credit-charging argument used to prove Theorem \ref{thmEp}
no longer suffices for proving Theorem \ref{thmEpmem}. Nonetheless, by
splitting the problem into two separate amortization arguments we are
able to complete the analysis. This is done in Lemma
\ref{lemamortizedcostsmemefficient}.

\begin{lemma}
  Consider a series-parallel graph $G = (V, E)$ recursively built from
  series and multi-spawn parallel combinations. Denote each
  multi-spawn combination by its tuple $(x_1, x_2, x_3, x_3, \ldots,
  x_t)$, where the odd-indexed $x_i$'s represent the $a_i$
  components and the even-indexed $x_i$'s represent the $b_i$
  components. Then
  \begin{equation}
    \sum_{(x_1, \ldots, x_t) \in \mathcal{M}} \sum_{i = 1}^{t - 1} \min\left(\sum_{j = 1}^{\lfloor i / 2 \rfloor} s(x_{2j}), p\right) \cdot s(x_i) \le O(|E| \cdot p),
    \label{eqmultispawncosts}
  \end{equation}
  where the set $\mathcal{M}$ contains all multi-spawn combinations in
  the recursive construction of $G$.
  \label{lemamortizedcostsmemefficient}
\end{lemma}

\begin{proof}

  We think of each multi-spawn combination $(x_1, \ldots, x_t)$ as
  consisting of $t - 1$ sub-combinations, in which after $i$
  sub-combinations we have combined $x_1, \ldots, x_{i + 1}$. One
  should think of the cost of the $(i - 1)$-th sub-combination is
  $$\min\left(\sum_{j = 1}^{\lfloor i / 2 \rfloor} s(x_{2j}), p\right) \cdot s(x_i).$$
  We say the sub-combination is \defn{heavy} if $s(x_i) = p$ and is
  \defn{light} if $s(x_i) < p$.

  The light sub-combinations can be handled using a similar argument
  as for the non-fully-formed case in Lemma
  \ref{lemamortizedcosts}. At the beginning of the recursive
  construction of $G$, assign to each edge $2p - 1$ credits. For each
  light sub-combination, combining some $x_1, \ldots, x_{i - 1}$
  with some $x_i$, we deduct $\min\left(\sum_{j = 1}^{\lfloor i / 2
    \rfloor} s(x_{2j}), p\right)$ credits from each edge in
  $x_i$. Since the number of edges in $x_i$ is at least $s(x_i)$, the
  sub-combination deducts a total of at least
  $\min\left(\sum_{j = 1}^{\lfloor i / 2 \rfloor} s(x_{2j}),
  p\right) \cdot s(x_i)$ credits. In order to bound the contribution
  of light sub-combinations to \eqref{eqmultispawncosts}, it
  suffices to prove that each edge $e \in E$ has a total of at most
  $2p - 1$ credits deducted from it. This follows by the exact same
  invariant-based argument as used in the proof of Lemma
  \ref{lemamortizedcosts}.

  In order to analyze the contribution of the heavy sub-combinations
  to \eqref{eqmultispawncosts}, we introduce a second amortization
  argument. Again we assign credits to edges, this time giving $p$
  credits to each edge $e$. As we recursively construct $G$ through
  series and multi-spawn combinations, we assign to each component
  $C$ a set of up to $p$ \defn{representative edges}, which includes
  all of $C$'s edges when $C$ contains $p$ or fewer edges, and $p$
  edges otherwise. When a component $C$ is constructed by combining
  two components $C_1$ and $C_2$ in series, $C$'s representative
  edges are the union of $C_1$'s and $C_2$'s (truncated to at most $p$
  edges). When a component $C$ is constructed by a multi-spawn
  combination $(x_1, \ldots, x_t)$ such that at least one of the
  $x_i$'s contains $p$ or more edges, $C$'s representative edges are
  inherited from the first such $x_i$; if none of the $x_i$'s contain
  $p$ or more edges, $C$'s representative edges are the union of the
  representative edges for each of $x_i$'s (truncated to at most $p$
  edges).

  Now consider a heavy sub-combination between sub-components $x_1,
  \ldots, x_{i - 1}$ and $x_i$ (recall that since the subcombination
  is heavy, we have that $s(x_i) = p$). If each of $x_1, \ldots, x_{i
    - 1}$ contains fewer than $p$ edges, then we deduct $p$ credits
  from each representative edge in each of $x_1, \ldots, x_{i -
    1}$. (Note that this is actually all of the edges in $x_1, \ldots,
  x_{i - 1}$.) If at least one of $x_1, \ldots, x_{i - 1}$ contains
  $p$ or more edges, then we deduct $p$ credits from each
  representative edge of $x_i$. In both cases, we deduct at least
  $\min\left(\sum_{j = 1}^{\lfloor i / 2 \rfloor} s(x_{2j}), p\right)
  \cdot s(x_i)$ credits in total, corresponding with the work done
  during the sub-combination.

  The deductions of credits are designed so that two important
  properties hold: (1) whenever an edge $e$ has credits deducted
  during a multi-spawn combination, the edge $e$ will no longer be a
  representative edge in the new component $C$ constructed by the
  multi-spawn combination; and (2) within a multi-spawn combination
  each edge $e$ will have credits deducted from it at most
  once. Combined, these properties ensure that each edge has credits
  deducted at most once during the full construction of $G$. This, in turn,
  ensures that the total number of credits deducted by the algorithm
  is at most $|E| \cdot p$, and that the contribution of heavy
  sub-combinations to \eqref{eqmultispawncosts} is also at most $|E|
  \cdot p$.

\end{proof}

\section{An Offline Approximate-Threshold Algorithm}\label{secappendixsimpleapprox}

In this section we present our algorithm for the approximate
threshold problem in the simpler offline setting, in which rather than
supporting multi-spawn combinations, our recursive algorithm needs
only support series and parallel combinations.

Our offline algorithm for the approximate threshold problem, which we
call \proc{ApproxOff}, will compute the high-water mark $h$ over a
special class of antichains that satisfy a certain property that we
call robustness. (This is similar to the notion of stripped robustness
from \secref{approx}, except without any requirements about
non-critical components; there are also several other minor
differences designed to yield the simplest possible final algorithm.)
The return value of the algorithm will then be determined by whether
the computed value $h$ is greater than $M/2$. In this section, we
define what it means for an antichain to be robust and prove the
correctness of our algorithm. Then, in Section \ref{secrobustcompute},
we describe a recursive procedure for computing the quantity $h$
needed by the algorithm in linear time $O(|E|)$.

When considering an antichain $A = (x_1, \ldots, x_q)$, we partition
the predecessors of the antichain, $\{e \mid e < x_i \text{ for some
}i\}$, into two categories. The \defn{core predecessors} $\mathcal{C}(A)$ of the antichain
$A$ is the set of edges that are predecessors to more than one member
of the antichain,
$$\mathcal{CP}(A) = \{e \mid e < x_i, e < x_j \text{ for some }i \neq j\}.$$ If
an edge $e$ is a predecessor of $A$ but not a core predecessor, then
$e$ is a \defn{local predecessor} of some $x_i$. We denote the set
of local predecessors of $x_i$ by
$$\mathcal{LP}_A(x_i) = \{e \mid e < x_i \text{ and } e \not< x_j
\forall j \neq i\}.$$

We define the \defn{core companions} $\mathcal{CC}(A)$ of the antichain $A$ to be the set of edges $e$ contained in a parallel component
$T_1$ with positive edge sum and whose partnering parallel component $T_2$ contains multiple
edges from the antichain $A$.  For each $x_i$, we define the
\defn{local companions} $\mathcal{LC}_A(x_i)$ of $x_i$ to be the set
of edges $e$ in a parallel component $T_1$ with positive edge sum and whose partnering parallel component $T_2$ contains the edge $x_i$ but not any other edge $x_j \in A$.

The \defn{core water mark} $C(A)$ is the sum of the edge totals over all edges in the core predecessors and companions of $A$,
$$C(A) = \sum_{e \in \mathcal{CP}(A) \cup \mathcal{CC}(A)} t(e).$$
Similarly, the \defn{local water mark} $L^\circ_A(x_i)$ of each edge $x_i \in A$ is given by
$$L^\circ_A(x_i) = m(x_i) + \sum_{e \in \mathcal{LP}_A(x_i) \cup \mathcal{LC}_A(x_i)} t(e).$$

The total water mark of the
antichain can be rewritten as
\begin{equation*}
  W(A) = C(A) + \sum_{i = 1}^q L^\circ_A(x_i).
  \label{eqwatermarkrewritten}
\end{equation*}

The \proc{ApproxOff} algorithm for the approximate threshold problem
will compute the infinite-processor high-water mark, except restricted
only to antichains $A$ whose local water marks all exceed
$\frac{M}{2p}$. We call an antichain $A = (x_1, \ldots, x_q)$
\defn{robust} if $L^\circ_A(x_i) > \frac{M}{2p}$ for each edge
$x_i$. The \defn{$p$-processor robust memory high-water mark}
$\HH_p^\circ(G)$ is defined to be
$$\HH^\circ_p(G) = \max_{A \in \mathcal{R}, \ |A| \le p} W(x_1, \ldots, x_q),$$
where $\mathcal{R}$ is the set of robust antichains in $E$.

The first step in our algorithm will be to compute the
infinite-processor robust memory high-water mark
$\HH_\infty^\circ(G)$. Then, if $\HH_\infty^\circ(G)
\le M / 2$, our algorithm will return 0, and if
$\HH_\infty^\circ(G) > M / 2$, our algorithm will return 1.

Computing $\HH_\infty^\circ(G)$ can be done in linear time $O(|E|)$
using a recursive algorithm described in Section
\ref{secrobustcompute}. The computation is made significantly easier,
in particular, by the fact that it is permitted to consider the
infinite-processor case rather than restricting to $p$ processors or
fewer.

On the other hand, the fact that $\HH_\infty^\circ(G)$ should
tell us anything useful about $\HH_p(G)$ is non-obvious. In the
rest of this section, we will prove the following theorem, which
implies the correctness of the algorithm:

\begin{theorem} \
   \begin{itemize}
  \item If $\HH_\infty^\circ(G) \le M / 2$, then $\HH_p(G)
    \le M$.
  \item If $\HH_\infty^\circ(G)  > M / 2$, then $\HH_p(G) > M / 2$.
  \end{itemize}
\label{thmalgcorrect}
\end{theorem}

To prove Theorem \ref{thmalgcorrect}, we begin by comparing
$\HH_p^\circ (G)$ and $\HH_p(G)$:
\begin{lemma} $$\HH_p^\circ (G) \ge \HH_p(G) -
\frac{M}{2}.$$
\label{lemremovenonrobust}
\end{lemma}
\begin{proof} Consider an antichain $A = (x_1, \ldots, x_q)$,
  with $q \le p$, that is not robust. One might try to construct a
  robust antichain $B$ by removing each $x_i \in A$ that satisfies
  $L^\circ_A(x_i) \le \frac{M}{2p}$. The removal of these $x_i$'s, however,
  would change the sets of local predecessors and companions for the remaining
  $x_j$'s, making it so that the new antichain $B$ may still not be
  robust.

  One can instead obtain a robust antichain through a more iterative
  approach. Begin with the antichain $A_1 = A$ that is not robust.
  Since $A_1$ is not robust, some $x_i \in A$ satisfies $L^\circ_A(x_i) \le
  \frac{M}{2p}$. Define $A_2$ to be the same antichain with $x_i$
  removed. If the antichain $A_2$ is also not robust, then pick some
  edge $x_j \in A_2$ such that $L^\circ_{A_2}(x_j) \le \frac{M}{2p}$, and
  define $A_3$ to be $A_2$ with $x_j$ removed. Continue like this
  until we reach some $A_r$ that is robust. (Note that one legal
  option for $A_r$ is the empty antichain, which is considered to be
  robust.)

  For any two consecutive antichains $A_i$ and $A_{i + 1}$ in the
  sequence, that differ by the removal of an edge $x_j$, the
  water marks satisfy
  \begin{equation}
    W(A_{i + 1}) \ge W(A_i) - L^\circ_{A_i}(x_j).
    \label{eqedgeremoval}
  \end{equation}
  (Note that the reason that \eqref{eqedgeremoval} is not true with
  equality is simply that the removal of $x_j$ may allow for the
  addition of a new companion component to the antichain $A_{i +
    1}$, thereby making $W(A_{i + 1})$ greater than $W(A_i) -
  L^\circ_{A_i}(x_j)$.)
  
  Since we only remove edges $x_j$ satisfying $L^\circ_{A_i}(x_j) \le \frac{M}{2p}$, it follows that
  $$W(A_{i + 1}) \ge W(A_i) - \frac{M}{2p}.$$ Moreover, in the
  processes of constructing the robust antichain $A_r$, we can remove
  a total of at most $q$ edges from the original antichain $A = (x_1,
  \ldots, x_q)$. Thus
  $$W(A_{r}) \ge W(A) - q \cdot \frac{M}{2p} \ge W(A) - \frac{M}{2}.$$
  This, in turn, implies that $\HH_p^\circ (G) \ge \HH_p(G) -
  \frac{M}{2}$, as desired.
\end{proof}

The following corollary proves the first part of Theorem \ref{thmalgcorrect}
\begin{corollary} If $\HH_\infty^\circ(G) \le M / 2$, then
$\HH_p(G) \le M$.
\label{correturn0}
\end{corollary}
\begin{proof}
If
$\HH_\infty^\circ(G) \le M / 2$, then $\HH_p^\circ(G)
\le M / 2$, and thus by Lemma \ref{lemremovenonrobust}, $\HH_p(G) \le M$.
\end{proof}

The second half of Theorem \ref{thmalgcorrect} is given by Lemma
\ref{lemreturn1}:
\begin{lemma} If $\HH_\infty^\circ(G) > M / 2$, then
$\HH_p(G) > M / 2$.
\label{lemreturn1}
\end{lemma}

\begin{proof} Since
$\HH_\infty^\circ(G) > M / 2$, one of the following must be
true:
\begin{itemize}
\item \textbf{There is a robust antichain $A = (x_1, \ldots, x_q)$
  with $q \le p$ such that $W(A) > M / 2$: }In this case, we trivially
  get that $\HH_p(G) > M / 2$.

\item \textbf{There is a robust antichain $A = (x_1, \ldots, x_q)$
  with $q > p$ such that $W(A) > M / 2$: }This case is somewhat more
  subtle, since the large number of edges in the antichain $A$ could
  cause $W(A)$ to be much larger than $\HH_p(G)$. We will use the
  robustness of $A$ in order to prove that the potentially much
  smaller antichain $B = (x_1, \ldots, x_p)$ still has a large water
  mark $W(B) > \frac{M}{2}$.

  Let $\mathcal{T}$ denote the set of edges $e \in E$ such that either
  $e \le x_i$ for some $i \in [p]$, or $e$ is contained in a companion
  parallel component of $B$. The quantity $W(B)$ can be written as
  \begin{align*}
    W(B)  &= \sum_{i = 1}^p m(x_i) + \sum_{e \in \mathcal{T}} t(e) \\
          &= \sum_{i = 1}^p L^\circ_A(x_i) + \sum_{e \in \mathcal{T} \cap (\mathcal{CP}(A) \cup \mathcal{CC}(A))} t(e).\end{align*}

  By the robustness of $A$, each local water mark $L^\circ_A(x_i)$ is
  greater than $\frac{M}{2p}$. Thus
  $$W(B) > \frac{M}{2} + \sum_{e \in \mathcal{T} \cap (\mathcal{CP}(A) \cup \mathcal{CC}(A))} t(e).$$

  Recall the downset-non-negativity property, which requires that
  every downset $S \subseteq E$ (meaning that the predecessors of the
  edges in $S$ are all in $S$) satisfy $\sum_{e \in S} t(e) \ge 0$. To
  see that the $\mathcal{T}$ is a downset, observe that it consists of
  two parts, the set $\mathcal{T}_1$ of predecessors of $B$, and the
  set $\mathcal{T}_2$ of edges contained in companion parallel
  components to $B$; since the set $\mathcal{T}_1$ is a downset,
  and because the predecessors of edges in $\mathcal{T}_2$ are all
  either in $\mathcal{T}_2$ or in $\mathcal{T}_1$, the full set
  $\mathcal{T} = \mathcal{T}_1 \cup \mathcal{T}_2$ is a
  downset. Similarly we claim that $\mathcal{CP}(A) \cup
  \mathcal{CC}(A)$ is a downset; in particular, $\mathcal{CP}(A)$ is a
  downset by its definition, and the predecessors of edges in
  $\mathcal{CC}(A)$ are all either contained in $\mathcal{CC}(A)$ or
  in $\mathcal{CP}(A)$. Since we have shown that $\mathcal{T}$ and
  $\mathcal{CP}(A) \cup \mathcal{CC}(A)$ are downsets, and because the
  intersection of two downsets is necessarily also a downset, it
  follows that $\mathcal{T} \cap (\mathcal{CP}(A) \cup
  \mathcal{CC}(A))$ is a downset.

  Applying the downset-non-negativity property, we get that
$$\sum_{e \in \mathcal{T} \cap (\mathcal{CP}(A) \cup \mathcal{CC}(A))} t(e) \ge 0,$$
  implying that $W(B) > \frac{M}{2}$, and completing the proof.
\end{itemize}
\end{proof}

\subsection{Computing $\HH_\infty^\circ(G)$ in linear time}\label{secrobustcompute}

In this section, we present a recursive algorithm for computing the
infinite-processor robust high-water mark $\HH_\infty^\circ(G)$ in
linear time $O(|E|)$. This, in turn, can be used within the
\proc{ApproxOff} algorithm to solve the approximate-threshold problem
in linear time $O(|E|)$.  We assume that we are given the
series-parallel DAG $G$, along with the labels $t(e)$ and $m(e)$ for
each $e \in G$.



Suppose we recursively build $G$ from series and parallel
combinations. Whenever we create a new component $C$ (by combining
two old ones) we will maintain the following information on the
component:
\begin{itemize}
\item The total allocation and freeing work done in $C$,
  $$\memtotal = \sum_{e \in C} t(e).$$
\item The memory high-water mark  with one processor,
  $$\maxchain = \HH_1(C).$$
\item The infinite-processor memory high-water mark
  restricted only to robust antichains containing more than one
  edge:
  $$\maxrobust = \max_{A \in \mathcal{R}, \ |A| > 1} W(A),$$ where
  $\mathcal{R}$ is the set of robust antichains in $G$. If $C$
  contains no multi-edge robust antichains, then $\maxrobust = \NULL$,
  and is treated as $-\infty$.
\end{itemize}

 The special handling in the recursion of antichains with only one
 edge (i.e., by $\maxchain$) is necessary because the local
 contribution of that edge $x$ is not yet completely determined until
 at least one other edge is added to the antichain. On the other hand,
 once a robust antichain contains multiple edges, the local
 contribution of each edge is now fixed, even if we combine this
 antichain with other antichains as we recursively construct the
 graph. This allows for all multi-edge robust antichains to be
 grouped together in the variable $\maxrobust$.


As a base case, for a component $C$ consisting of a single edge
$e$, we initialize the variables as follows: $\memtotal = t(e)$,
$\maxchain = m(e)$, and $\maxrobust = \NULL$.

When we combine two components $C_1$ and $C_2$ in series to build a new component $C$, the three variables can be updated as follows:
\begin{itemize}
\item We update $C.\memtotal$ as
  $$C_1.\memtotal + C_2.\memtotal.$$
  In particular, $\sum_{e \in C}t(e) = \sum_{e \in C_1} t(e) + \sum_{e \in C_2} t(e)$.
\item We update $C.\maxchain$ as
  $$\max(C_1.\maxchain, C_1.\memtotal +
  C_2.\maxchain).$$ In particular, every single-edge antichain in
  $C_1$ has the same water mark in $C$ as it did in $C_1$, and every
  single-edge antichain in $C_2$ has cost in $C$ an additional
  $C_1.\memtotal$ greater than it did in $C_2$.
\item We update $C.\maxrobust$ as
   \[
    \max(C_1.\maxrobust, C_1.\memtotal + C_2.\maxrobust).
  \]

  In particular, the set of multi-edge robust
  antichains in the new component $C$ is the union of the set of
  multi-edge antichains in $C_1$ with the set of multi-edge antichains
  in $C_2$. Whereas each of the multi-edge antichains in $C_1$ have the
  same water mark in $C$ as they did in $C_1$, the multi-edge antichains
  in $C_2$ each have their water marks increased by $C_1.\memtotal$.
\end{itemize}

When we combine two components $C_1$ and $C_2$ in parallel to build a new component $C$, the three variables can be updated as follows:
\begin{itemize}
\item We update $C.\memtotal$ as
  $$C_1.\memtotal + C_2.\memtotal.$$
  In particular, just as before, $\sum_{e \in C}t(e) = \sum_{e \in C_1} t(e) + \sum_{e \in C_2} t(e)$.
\item We update $C.\maxchain$ as
  \begin{equation*}
    \begin{split}
  \max(C_1.\maxchain + \max(0, C_2.\memtotal),  \\
  C_2.\maxchain + \max(0, C_1.\memtotal)).
    \end{split}
  \end{equation*}
  In particular, the
  set of single-edge antichains in $C$ is the union of the set of
  single-edge antichains in $C_1$ with the set of single-edge
  antichains in $C_2$. Since the water marks of the antichains are the
  same in $C_1$ and $C_2$ as they are in $C$, except with the addition
  of $\max(0, C_2.\memtotal)$ and $\max(0, C_1.\memtotal)$
  respectively (due to the possibility of $C_2$ and $C_1$ being
  suspended companion parallel components), $C.\maxchain$ can be
  updated by taking a simple maximum of the two options.
\item
  The update of $C.\maxrobust$ is slightly more subtle. Define
  $C_1.\overline{\maxchain}$ and $C_2.\overline{\maxchain}$ to be the
  highest water marks achieved by robust single-edge antichains in
  $C_1$ and $C_2$, respectively. That is,
  $$C_1.\overline{\maxchain} =
  \begin{cases}
    C_1.\maxchain \text{ if }C_1.\maxchain > \frac{M}{2p} \\
    \NULL \text{ otherwise,}
  \end{cases}
  $$
  and
  $$C_2.\overline{\maxchain} =
  \begin{cases}
    C_2.\maxchain \text{ if }C_2.\maxchain > \frac{M}{2p} \\
    \NULL \text{ otherwise.}
  \end{cases}
  $$

  Define $\mathcal{R}(C)$, $\mathcal{R}(C_1)$, and $\mathcal{R}(C_2)$ to be
  the sets of robust antichains in $C$, $C_1$, and $C_2$,
  respectively. Then, because $C$ is obtained by combining $C_1$ and
  $C_2$ in parallel,
  $$\mathcal{R}(C) = \{x \cup y \mid x \in \mathcal{R}(C_1), y \in
  \mathcal{R}(C_2)\}.$$

  When computing $C.\maxrobust$, we are interested exclusively in the
  antichains $A$ satisfying $|A| > 1$. If $x = A \cap C_1$ and $y = A
  \cap C_2$, then the requirement that $|A| > 1$ translates into the
  requirement that (at least) one of the following three requirements
  holds:
  \begin{enumerate}
  \item $|x| = |y| = 1$: The maximum water mark for robust antichains $a$ such that $|x| = |y| = 1$
  is given by
  $$C_1.\overline{\maxchain} + C_2.\overline{\maxchain}.$$
  \item $|y| > 1$: The maximum water mark for robust antichains $a$ such that  $|y| > 1$ is given by
    \begin{multline*}
      \max(C_1.\overline{\maxchain}, \\ C_1.\maxrobust, C_1.\memtotal, 0) \\
      + C_2.\maxrobust,
    \end{multline*}
    where entries in the maximum correspond
    with the cases where $|x| > 1$; $|x| = 1$; $|x| = 0$ and $C_1$ is
    included as a suspended companion component; and $|x| = 0$ and $C_1$ is
    not included as a suspended companion component.
  \item $|x| > 1$: The maximum water mark for robust antichains $a$ such that  $|x| > 1$ is given by
    \begin{multline*}
      \max(C_2.\overline{\maxchain}, \\ C_2.\maxrobust, C_2.\memtotal, 0) \\ + C_1.\maxrobust,
    \end{multline*}
    where entries in the maximum correspond with the cases where $|y|
    > 1$; $|y| = 1$; $|y| = 0$ and $C_2$ is included as a suspended
    companion component; and $|y| = 0$ and $C_2$ is not included as
    a suspended companion component.
  \end{enumerate}
  Combining the cases, we can update $C.\maxrobust$ as
  \[\begin{array}{ll}
 \max \Big( & C_1.\overline{\maxchain} + C_2.\overline{\maxchain}, \\
                            & \max(C_1.\overline{\maxchain}, C_1.\maxrobust, \\ & C_1.\memtotal, 0) + C_2.\maxrobust, \\
  & \max(C_2.\overline{\maxchain}, C_2.\maxrobust, \\ & C_2.\memtotal, 0) + C_1.\maxrobust \Big).
  \end{array}
  \]
\end{itemize}

Using the recursive construction described above, we can compute the
variables $\memtotal$, $\maxchain$, and $\maxrobust$ for our graph $G$ in
linear time $O(|E|)$. In order to then compute
$\HH_\infty^\circ(G)$, the infinite-processor high-water mark
considering only robust antichains, we simply compute
\begin{multline*}
\HH_\infty^\circ(G) =
\max\Big(  
\left\{
  \begin{array}{l}
    \maxchain \text{ if }\maxchain > \frac{M}{2p} \\
    \NULL \text{ otherwise}
  \end{array}\right\}
  ,\\ \maxrobust, 0
  \Big).
\end{multline*}

  \section{Recursing on Multi-Spawn Components}\label{secmultispawnrecursion}

  In this section, we complete the recursion for
  $\HH_{\infty}^\bullet(G)$ discussed in Section
  \ref{sec:strippedrobustcompute} by handling the case of multi-spawn
  combinations.

 Consider a multi-spawn combination $C$ as in Figure
 \ref{figmultispawns} with components $C_1 = a_0, C_2 = b_1, C_3 =
 a_1, \ldots, C_{2k + 1} = a_k$.

  Throughout the section, we will use the notation $m(b_i)$ and
  $t(b_i)$ introduced in Lemma \ref{lemcandidateantichains}. When
  $b_i$ is the first case of the lemma, we say that \defn{$b_i$ is a
    natural companion} and that $b_i$'s \defn{natural contribution} is
  $t(b_i)$; when $b_i$ is in the second case, we say that \defn{$b_i$
    is naturally dormant} and that $b_i$'s \defn{natural contribution}
  is $0$; when $b_i$ is in the third case, we say that \defn{$b_i$ is
    naturally active} and that $b_i$'s \defn{natural contribution} is
  $m(b_i)$.

Recall that the execution of the parallel program on one thread
computes the recursive values for each $C_i$ with $i$ iterating
through the range $i = 1, \ldots, 2k + 1$. We wish to use these in
order to compute the recursive values for $C$.

 To do this, we maintain a collection of intermediate values during
 the execution of the components $C_1 , \ldots, C_{2k + 1}$. Before
 introducing these intermediate values, we define a few terms.

 We call a stripped robust antichain $A$ in $C$ a \defn{candidate}
 antichain if for each $b_i$ in $C$ such that $A$ contains an edge in
 one of $a_i, b_{i + 1}, a_{i + 1}, \ldots$, the three properties
 stated in Lemma \ref{lemcandidateantichains} hold. (In particular,
 $b_i$'s local contribution $L^\bullet_A(b_i)$ should be precisely
 $b_i$'s natural contribution.)  By Lemma
 \ref{lemcandidateantichains}, when computing computing
 $C.\maxrobust$, it suffices to consider only multi-edge candidate
 antichains.
 
 In order to describe the intermediate values that we maintain during
 the execution of the components, we will also need the notion of a
 suspendend-end and ignored-end water mark. (These are the same
 definitions as used in Appendix \ref{secappendixonlineexact}.) If an
 antichain $A$ in $C$ contains only edges in $b_1, \ldots, b_{k}$, and
 $b_t$ is the largest $t$ such that $b_t$ contains a edge in $A$, then
 we say $A$ has a \defn{suspended end} if the components $a_{t + 1},
 b_{t + 2}, \ldots, a_{k}$ form a companion component of $A$ (which
 occurs if the sum of their edge costs is net positive).  We cannot
 know whether $A$ will have a suspended end until we have seen all of
 $a_{t + 1}, b_{t + 2}, \ldots, a_k$ (i.e., until we have completed
 the entire multi-spawn component). Thus, for each antichain $A$ that
 contains only edges in $b_1, \ldots, b_k$, we consider both a
 ``suspended-end'' version of its water mark and a ``ignored-end''
 version of its water mark.  The \defn{suspended-end water mark} of
 $A$ is $W(A)$ if $A$ has a suspended end, and is $W(A) + \sum_{e \in
   a_{t + 1}, b_{t + 2}, \ldots, a_{k}} t(e)$ if $A$ does not have a
 suspended end (i.e., it is the water mark $A$ would have if it had a
 suspended end). Similarly, the \defn{ignored-end water mark} of $A$
 is $W(A)$ if $A$ does not have a suspended end, and is $W(A) -
 \sum_{e \in a_{t + 1}, b_{t + 2}, \ldots, a_{k}} t(e)$ if $A$ does
 have a suspended end. These definitions will prove useful when
 defining the intermediate values maintained by our
 algorithm. Additionally, for any antichain $A$ that contains an edge
 in some $a_t$, we define the \defn{ignored-end water mark} of $A$ to
 be the water mark of $A$; thus the ignored-end water mark is defined
 for all antichains in $A$ of $C$.

 After having executed each of $C_1, \ldots, C_l$, let $l_1$ be the
 index of the largest-indexed $a_i$ executed and $l_2$ be the index of
 the largest-indexed $b_i$ executed. We maintain the following
 intermediate values:
 \begin{itemize}
    \item $\multirobustsuspendend_l$: This is the maximum
      suspended-end cost of any multi-edge candidate antichain $A$ in
      $C_1 \cup \cdots \cup C_l$ containing only edges in $b_1,
      \ldots, b_{l_2}$. If no such $A$ exists, this is $\NULL$. Note
      that here $C_1 \cup \cdots \cup C_l$ is treated as a multi-spawn
      component and the costs of the antichains are considered just
      within the graph $C_1 \cup \cdots \cup C_l$, rather than the
      full graph $C$ (which matters because we are considering the
      suspended-end cost of the antichain).

 \item $\multirobustignoreend_l$: This is the maximum ignored-end cost
   of any multi-edge candidate antichain $A$ in $C_1 \cup \cdots \cup
   C_l$. If no such $A$
   exists, this is $\NULL$.

 \item $\singlesuspendend_l$: This is the maximum suspended-end cost
   of any single-edge antichain $A$ in $C_1 \cup \cdots \cup C_l$ such
   that $A$ contains only edges in $b_1, \ldots, b_{l_2}$. (Again, we
   consider the suspended-end cost just within the graph $C_1 \cup \cdots \cup
   \cdots \cup C_l$.)
 \item $\singleignoreend_l$: This is the maximum ignored-end cost of
   any single-edge antichain $A$ in $C_1 \cup \cdots \cup C_l$.

 \item $\robustunfinished_l$: Let $t$ be the largest $t \le l_2$ such that
   $b_t$ is naturally active. Then $\robustunfinished_l$ is the sum of
   the natural contributions of $b_1,\ldots, b_t$, along with $t(a_0),
   t(a_1) + \cdots + t(a_{t - 1})$. If no such $t$ exists, then
   $\robustunfinished_l$ is $\NULL$.

   One should think of this as the contribution of $b_1, \ldots, b_t$
   and $a_0, \dots, a_{t - 1}$ to any candidate antichain in $C$ that
   contains at least one edge in $b_{l_2 + 1}, \ldots, b_k$ or $a_{l_1
     + 1}, \ldots, a_k$. (We separate this from the contribution of
   the edges $b_{t + 1}, \ldots, b_{l_2}$ and $a_t, \ldots, a_{l_1}$
   which are considered by the next quantity.)
 \item $\robustunfinishedtail_l$: Let $t$ be the largest $t \le l$
   such that $b_t$ is naturally active, or $0$ if no such $t$
   exists. Then $\robustunfinishedtail_l$ is the sum of the natural
   contributions of $b_{t + 1}, \ldots, b_{l_2}$, along with $t(a_t) +
   t(a_{t + 1}) + \cdots + t(a_{l_1})$.

   One should think of this as the contribution of $b_{t + 1}, \ldots,
   b_{l_2}$ and $a_t, \dots, a_{l_1}$ to any candidate antichain in
   $C$ that contains at least one edge in $b_{l_2 + 1}, \ldots, b_k$
   or $a_{l_1 + 1}, \ldots, a_k$. The quantity
   $\robustunfinishedtail_l$ is handled separately from
   $\robustunfinished_l$ because if the candidate antichain contains
   only a single edge in $b_{l_2 + 1}, \ldots, b_k$ or $a_{l_1 + 1},
   \ldots, a_k$, then $\robustunfinishedtail_l$ can affect the local
   contribution of that edge.
 \item $\edgetotal_l$: This is $\sum_{i = 0}^{l_1} t(a_i) + \sum_{i = 1}^{l_2}
   t(b_i)$, the total sum of the edge totals over all edges in the
   components $a_0, \ldots, a_{l_1}$, $b_1, \ldots, b_{l_2}$. 
 \item $\emptytail_l$: This is $\sum_{i = 0}^{l_1} t(a_i) + \sum_{i = 1}^{l_2}
   \max(0, t(b_i))$. One should think of this as the contribution of
   $a_0, \ldots, a_{l_1}$, $b_1, \ldots, b_{l_2}$ to any single-edge
   antichain in $C$ whose edge lies in one of $a_{l_1 + 1}, a_{l_1 +
     2}, \ldots$ or $b_{l_2 + 1}, b_{l_2 + 2}, \ldots$.
 \end{itemize}

 Given the above variables for $l = 2k + 1$, one can compute $$C.\memtotal =
 \edgetotal_{2k + 1},$$
 \begin{multline*}
   C.\maxrobust = \max(\multirobustsuspendend_{2k + 1}, \\
   \multirobustignoreend_{2k + 1}),
 \end{multline*}
 and
 \begin{multline*}
   C.\maxchain =
   \max(\singlesuspendend_{2k + 1}, \\ \singleignoreend_{2k + 1}).
 \end{multline*}

 Prior to beginning, we have $l = 0$, and $\multirobustsuspendend_0 = \NULL$,
 $\multirobustignoreend_0 = \NULL$, $\singlesuspendend_0 = \NULL$,
 $\singleignoreend_0 = \NULL$, $\robustunfinished_0 = \NULL$,
 $\robustunfinishedtail_0 = 0$, $\edgetotal_0 = 0$, and $\emptytail_0
 = 0$.
 
To complete the algorithm, we present the protocol for advancing $l$
by one, and updating each of the intermediate values.

Suppose for some odd $l > 0$ we are given the values of the above
quantities for $l - 1$, and given the recursive values for $a_{(l + 1) /
  2}$. We obtain the new values for $l$ as follows:
\begin{itemize}
\item \textbf{Step 1: Simple Updates. }
  We compute $\multirobustsuspendend_l$ as,
  $$\multirobustsuspendend_{l - 1} + a_{(l + 1)/2}.\memtotal,$$
  and $\singlesuspendend_l$ as,
  $$\singlesuspendend_{l - 1} + a_{(l + 1)/2}.\memtotal.$$

  We compute $\singleignoreend_l$ as
  \begin{multline*}
    \max(\singleignoreend_{l - 1}, \\a_{(l + 1)/2}.\maxchain + \emptytail_{l - 1}),
  \end{multline*}
  where the second entry in the maximum is the largest water mark of any single-edge antichain in $C$ with an edge in $a_{(l + 1) / 2}$.

  We set $\robustunfinished_l = \robustunfinished_{l - 1}$. Finally we
  increase each of $\robustunfinishedtail_l$, $\edgetotal_l$, and
  $\emptytail_l$ by $a_{(l + 1)/2}.\memtotal$ over their values for $l
  - 1$ (where the outcome is $\NULL$ if they were previously $\NULL$).
     
 \item \textbf{Step 2: Computing $\multirobustignoreend_l$.}   We update $\multirobustignoreend_l$ with Algorithm
  \ref{alg:updatemultirobustignoreend}. The only
  antichains $A$ that $\multirobustignoreend_{l}$ needs to consider
  but that $\multirobustignoreend_{l - 1}$ did not are the candidate
  stripped robust antichains $A$ containing at least one edge in
  $a_{(l + 1) / 2}$.

   The first if-statement checks whether any multi-edge candidate
   antichains exist in which $a_{(l + 1) / 2}$ contributes only a
   single edge; this requires that $\robustunfinishedtail_{l - 1} +
   a_{(l + 1) / 2}.\maxchain > \frac{M}{2p}$ in order for the local
   contribution of the edge in $a_{(l + 1) / 2}$ to exceed
   $\frac{M}{2p}$; and that $\robustunfinished \neq \NULL$ that way
   the resulting antichain contains multiple edges.

   The second if-statement considers candidate antichains in which
   $a_{(l + 1) / 2}$ contributes multiple edges. If
   $\robustunfinished_{l - 1} \neq \NULL$, then the maximum water mark
   in $C$ obtainable by such an antichain is $\robustunfinished_{l -
     1} + \robustunfinishedtail_{l - 1}$ $+ a_{(l + 1) /
     2}.\maxrobust$. If $\robustunfinished_{l - 1} = \NULL$, then the
   maximum water mark in $C$ obtainable by such an antichain is simply
   $\robustunfinishedtail_{l - 1} + a_{(l + 1) / 2}.\maxrobust$.
\end{itemize}

Suppose for some even $l > 0$ we are given the values of the intermediate
values for $l - 1$, and given the recursive values for $b_{l / 2}$. We
obtain the new values for $l$ as follows:
\begin{itemize}
\item \textbf{Step 1: Simple Updates: }
  We compute $\singlesuspendend_l$ as
  \begin{equation*}
    \begin{split}
      \max(\singlesuspendend_{l - 1}  + b_{l/2}.\memtotal, \\ b_{l/2}.\maxchain + \emptytail_{l - 1}),
    \end{split}
  \end{equation*}
  and $\singleignoreend_l$ as,
  $$\max(\singleignoreend_{l - 1}, b_{l/2}.\maxchain + \emptytail_{l - 1}).$$

We compute $\edgetotal_l$ as,
  $$\edgetotal_{l - 1} + b_{l/2}.\memtotal.$$
Finally, we compute $\emptytail_l$ as,
  $$\emptytail_l = \emptytail_{l - 1} + \max(0, b_{l/2}.\memtotal).$$

\item \textbf{Step 2: Computing $\multirobustsuspendend_l$ and $\multirobustignoreend_l$. }
  We update $\multirobustsuspendend_l$ and $\multirobustignoreend_l$
  with Algorithm \ref{alg:updatemultirobustsuspendend}. We begin by
  computing $X$, the largest ignored-end cost of any candidate
  stripped robust antichain in $C$ that (1) contains multiple edges;
  (2) contains at least one edge in $b_{l/2}$; and (3) contains no
  edges in $a_{l / 2}, b_{l / 2 + 1}, \ldots, a_k$. The first
  if-statement considers the case where the antichain has one edge in
  $b_{l/2}$; and the second considers the case where there are
  multiple such edges.

  After computing $X$, we update $\multirobustsuspendend_l$ and
  $\multirobustignoreend_l$ based on $X$'s value.
 
\item \textbf{Step 3: Computing $\robustunfinished_l$ and
  $\robustunfinishedtail_l$. }
   We compute $\robustunfinished_l$ and $\robustunfinishedtail_l$ with
  Algorithm \ref{alg:updaterobustunfinished}.
  We define $m$ and $t$ to be $m(b_{l / 2})$ and $t(b_{l /
    2})$, as defined in Lemma \ref{lemcandidateantichains}. We then
  update $\robustunfinished_l$ and $\robustunfinishedtail_l$
  appropriately based on the three cases in the lemma. (In the final
  case, we take the maximum of $0$ and $\robustunfinished_{l - 1}$
  because if the latter is $\NULL$, we wish to treat it as zero.)

\end{itemize}

This completes the recursion described in Section
\ref{sec:strippedrobustcompute}, allowing one to compute
$\HH_\infty^\bullet(G)$ in an online manner (i.e., while executing the
parallel program on a single thread) with constant multiplicative time
and space overhead.

   \begin{algorithm*}
     $\multirobustignoreend_{l} = \multirobustignoreend_{l - 1}$; \\
     \If{$\robustunfinishedtail_{l - 1} + a_{(l + 1) / 2}.\maxchain > \frac{M}{2p}$ and $\robustunfinished_{l - 1} \neq \NULL$}{
       $\multirobustignoreend_{l} = \max(\text{self},  \robustunfinished_{l - 1}$ $ + \robustunfinishedtail_{l - 1} + a_{(l + 1) / 2}.\maxchain)$; \\
     }
     \If{$a_{(l + 1) / 2}.\maxrobust \neq \NULL$}{
       \If{$\robustunfinished_{l - 1} \neq \NULL$}{
         $\multirobustignoreend_{l} = \max(\text{self}, \robustunfinished_{l - 1}$ $+ \robustunfinishedtail_{l - 1} + a_{(l + 1) / 2}.\maxrobust)$; \\
         }
       \If{$\robustunfinished_{l - 1} = \NULL$}{
         $\multirobustignoreend_{l} = \max(\text{self}, \robustunfinishedtail_{l - 1}$ $ + a_{(l + 1) / 2}.\maxrobust)$; \\
        }
     }
     \caption{Updating $\multirobustignoreend$ for $a_{(l + 1)/2}$}
     \label{alg:updatemultirobustignoreend}
   \end{algorithm*}

  \begin{algorithm*}
    $X = \NULL$; \\
    \If{$b_{l/2}.\maxchain + \robustunfinishedtail_{l - 1} > \frac{M}{2p}$ and $\robustunfinished_{l - 1} \neq \NULL$} {
      $X = b_{l/2}.\maxchain + \robustunfinishedtail_{l - 1}$ $ + \robustunfinished_{l - 1}$;\\
    }
    \If{$b_{l/2}.\maxrobust \neq \NULL$}{
      \If{$\robustunfinished_{l - 1} \neq \NULL$}{
        $X = \max(X, \robustunfinished_{l - 1} $ $+ \robustunfinishedtail_{l - 1} + b_{l/2}.\maxrobust)$
      }
      \If{$\robustunfinished_{l - 1} = \NULL$}{
        $X = \max(X, \robustunfinishedtail_{l - 1} + b_{l/2}.\maxrobust)$
      }
    }
    $\multirobustsuspendend_l = \max(X, \multirobustsuspendend_{l - 1})$; \\
    $\multirobustignoreend_l = \max(X, \multirobustignoreend_{l - 1})$; \\
    
     \caption{Updating $\multirobustsuspendend$ and $\multirobustignoreend$ for $b_{l/2}$}
     \label{alg:updatemultirobustsuspendend}
  \end{algorithm*}

  \begin{algorithm*}
    $m = b_{l/2}.\maxrobust$;\\
    \If{$b_{l/2}.\maxchain > \frac{M}{2p}$}{
      $m = \max(m, b_{l/2}.\maxchain)$; \\
    }
    \If{$m = \NULL$}{
      $m = 0$; \\
    }
    $t = b_{l/2}.\memtotal$; \\
    \If{$t > 0$ and $m \le t + \frac{M}{2p}$}{
      $\robustunfinished_l = \robustunfinished_{l - 1}$; \\
      $\robustunfinishedtail_l = \robustunfinishedtail_{l - 1} +  t$; \\
    }
    \If{$t \le 0$ and $m \le \frac{M}{2p}$}{
      $\robustunfinished_l = \robustunfinished_{l - 1}$; \\
      $\robustunfinishedtail_l = \robustunfinishedtail_{l - 1}$; \\
    }
    \If{$m \ge \max(0, t) + \frac{M}{2p}$}{
      $\robustunfinished_l = \max(0, \robustunfinished_{l - 1}) + \robustunfinishedtail_{l - 1} + m$; \\
      $\robustunfinishedtail_l = 0$; \\
    } 
     \caption{Updating $\robustunfinished$ and $\robustunfinishedtail$ for $b_{l/2}$}
     \label{alg:updaterobustunfinished}
   \end{algorithm*}

\fi

\end{document}